\documentclass[11pt]{article}
\usepackage[utf8]{inputenc}
\usepackage[T1]{fontenc}
\usepackage{graphicx}
\usepackage{longtable}
\usepackage{wrapfig}
\usepackage{rotating}
\usepackage[normalem]{ulem}
\usepackage{amsmath}
\usepackage{amssymb}
\usepackage{capt-of}
\usepackage{jheppub}
\usepackage{dsfont}
\usepackage[obeyFinal]{easy-todo}
\usepackage{tikz}
\usepackage{tikz-feynman}
\usepackage{mathtools}
\usepackage{amsthm}
\usepackage{subcaption}
\usepackage{cleveref}
\usetikzlibrary{shapes.geometric, arrows}
\hypersetup{colorlinks=true,linkcolor=blue}
\newtheorem{lemma}{Lemma}[section]
\newcommand{\eqnDiag}[1]{ \vcenter{\hbox{#1}} }

\tikzstyle{startstop} = [rectangle, rounded corners, 
minimum width=3cm, 
minimum height=1cm,
text centered, 
draw=black, 
fill=red!30]

\tikzstyle{process} = [rectangle, 
minimum width=2cm,
minimum height=1cm, 
text centered, 
draw=black, 
fill=orange!30]

\tikzstyle{process2} = [rectangle, 
minimum width=3cm, 
minimum height=1cm, 
text centered, 
text width=3cm, 
draw=black, 
fill=orange!30]

\tikzstyle{decision} = [rectangle, 
minimum width=3cm, 
minimum height=1cm, 
text centered, 
draw=black, 
fill=green!30]

\tikzstyle{caption} = [rectangle, 
minimum width=4cm, 
minimum height=1cm, 
text centered,
text width = 4cm]
\tikzstyle{arrow} = [thick,->,>=stealth]
\abstract{
We investigate a novel theoretical structure underlying the computation of
integration-by-parts relations between Feynman integrals via syzygy-based
methods.
Building on insights from intersection theory, we analyze the large-$\epsilon$
limit of dimensional regularization on the maximal cut, showing that total
derivatives vanish on the critical locus of the logarithm of the Baikov
polynomial--the locus known to govern the number of master integrals.
We introduce ``critical syzygies'' as a distinguished subset of syzygies that
captures this behavior. We show that, when the critical locus is isolated,
critical syzygies generate a sufficient set of total derivatives in the
large-$\epsilon$ limit.
We study their structure analytically at one loop and develop a numerical
approach for their construction at two loops.
Our results demonstrate that critical syzygies are a valuable tool for
integral reduction in cutting-edge two-loop examples, offering a novel geometric
perspective on integration-by-parts relations.
}
\author{Ben Page, Qian Song}
\affiliation{Department of Physics and Astronomy, Ghent University, 9000 Ghent, Belgium}

\emailAdd{ben.page@ugent.be, qian.song@ugent.be}
\date{\today}
\title{Critical Points and Syzygies for Feynman Integrals}
\begin{document}

\maketitle

\section{Introduction}

Feynman integrals are an essential component of precise predictions of scattering
observables at particle colliders. Moreover, they are increasingly important for making
predictions for gravitational wave experiments.
In cutting-edge calculations in perturbative quantum field theory, one of the
major bottlenecks is the large number of Feynman integrals that arise.
For this reason, in modern computational frameworks, one makes extensive use of
the fact that Feynman integrals exhibit linear relations with rational
coefficients.
These so-called ``integration by parts'' (IBP)
relations~\cite{Tkachov:1981wb,Chetyrkin:1981qh} are then used to reduce the
number of integrals to a much smaller number of ``master integrals''.
Moreover, these integral relations are essential for the calculation of Feynman
integrals, forming the backbone of the ``differential equations''
approach~\cite{Kotikov:1990kg,Remiddi:1997ny,Gehrmann:1999as,Henn:2013pwa}.

Given the large importance of relations between Feynman integrals in perturbative
calculations, it is perhaps not surprising that a large amount of effort has
gone into understanding them. The classical approach to integral reduction,
under the name of the ``Laporta'' algorithm~\cite{Laporta:2000dsw}, is to
construct a large set of total derivatives, which integrate to zero. One then
interprets these as relations between Feynman integrals and solves the relations
using linear algebra. This method has been very successful, and many public
implementations make use of it~\cite{
Maierhofer:2017gsa,Klappert:2020nbg,% Kira, Kira 2.0
Smirnov:2008iw,Smirnov:2019qkx,Smirnov:2023yhb,% Fire, FIRE6
Peraro:2019svx,% FiniteFlow
Guan:2024byi,% Blade
Lee:2012cn% LiteRed
}. In recent years, many such approaches have been able to reach new heights by
performing the reduction on numerical phase-space points using modular
arithmetic~\cite{vonManteuffel:2014ixa,Peraro:2016wsq,DeLaurentis:2022otd} and
reconstructing analytic results from these evaluations.
In tandem, a number of approaches have been introduced to organize the set of
total derivatives that one handles. One such approach is the ``block-triangular
form'', where a system of total derivatives is constructed which allows for
rapid numerical evaluations~\cite{Guan:2019bcx, Guan:2024byi}.
Another approach recently under study, is to organize the relations to form
symbolic reduction rules~\cite{Kosower:2018obg,Smith:2025xes}.
Moreover, there have been a number of recent investigations into the application of
machine-learning techniques to improving the Laporta algorithm~\cite{Song:2025pwy,vonHippel:2025okr,Zeng:2025xbh}.
Beyond this, there are the so-called ``syzygy-based''
approaches~\cite{Gluza:2010ws}, which avoid introducing auxiliary terms into the
collection of total derivatives that one constructs.
A further related approach, recently under study is that of constructing
parametric annihilators~\cite{Bertolini:2025zud}.
There have also been important developments in understanding mathematical
structures that control the IBP relations, in the language of ``intersection
theory''~\cite{Mastrolia:2018uzb, Frellesvig:2020qot}. Here, one makes use of
the ``intersection product'' that one can define on Feynman integrals in Baikov
representation to directly reduce integrals to master integrals.

In this work, we study the ``syzygy'' approach, motivated by its importance in the
two-loop numerical unitarity
method~\cite{Abreu:2017xsl,Abreu:2017hqn,Abreu:2018zmy,Abreu:2020xvt}.
Since its introduction in ref.~\cite{Gluza:2010ws}, the syzygy approach to
relations between Feynman integrals has received a great deal of study. Beyond
its original formulation in momentum space, it has been formulated in so-called
``adapted coordinates''~\cite{Ita:2015tya, Abreu:2017xsl}, embedding
space~\cite{Simmons-Duffin:2012juh,Caron-Huot:2014lda,Bern:2017gdk}, and,
prominently, the Baikov representation~\cite{Baikov:1996rk, Larsen:2015ped}.
Many methods of computing solutions of the syzygy problem have been proposed.
While some solutions are known in closed form~\cite{Ita:2015tya, Bohm:2017qme},
most approaches to the syzygy problem are computational in nature. Early
approaches made use of Groebner basis techniques based upon Schreyer's
theorem~\cite{cox2005using,Zhang:2016kfo} as implemented in computer algebra
packages such as Singular~\cite{DGPS}. More recent approaches make use of the
``module intersection'' strategy~\cite{Bohm:2018bdy} in Baikov representation.
Another prominent approach is to reduce the syzygy problem to that of solving linear systems~\cite{Schabinger:2011dz, CabarcasDing2011,
Agarwal:2020dye, Abreu:2023bdp}. By now, the problem of syzygy construction is
well enough understood that there exist public codes such as
\texttt{NeatIBP}~\cite{Wu:2023upw, Wu:2025aeg}.

In this work, we introduce a novel theoretical contribution to the understanding
of syzygies and Feynman integrals. Specifically, we uncover a deep connection
between these methods and recent applications of intersection
theory~\cite{Mizera:2017rqa,Mastrolia:2018uzb} to Feynman integrals.
In recent work~\cite{Mizera:2019vvs}, it was observed that the intersection
theory of dimensionally-regulated Feynman integrals greatly simplifies in the limit of large
dimension--the large-$\epsilon$ limit.
In the first part of our work, inspired by this observation, we study the
large-$\epsilon$ limit of total derivatives in the syzygy formalism. In this
limit, we observe that, when taken on the maximal cut, total derivatives vanish
on the critical locus of the logarithm of the Baikov polynomial: exactly the locus
studied in ref.~\cite{Mizera:2019vvs}.
We then formulate this observation in the language of algebraic geometry,
discussing how this geometric statement arises from theoretical considerations
of syzygies. Specifically, we highlight a connection between syzygies
and the ``ideal quotient'': a geometrical operation which corresponds to the
removal of branches of an algebraic variety.
This theoretical correspondence allow us to define a distinguished subset of
syzygies that we dub ``critical syzygies''--those singled out in the
large-$\epsilon$ limit. Appealing to the Lee-Pomeransky approach for counting
master integrals~\cite{Lee:2013hzt}, we then argue that, in cases where the
critical locus of the maximal cut of the Baikov is isolated, they give rise to a
complete collection of total derivatives in the large-$\epsilon$ limit.

In the rest of our work, we study this critical syzygy construction. We first
discuss how critical syzygies arise in the context of one-loop Feynman integrals
and how they are controlled by the geometry of the critical locus of the Baikov
polynomial. In particular, we directly show
how critical syzygies give rise to a complete set of total derivatives relevant
for gauge theories in cases where the critical locus of the logarithm of the
maximal-cut Baikov polynomial is isolated.
We then turn to the treatment of critical syzygies at two loops, where their
construction is mathematically much more complicated and we satisfy ourselves 
with computational studies.
An important point is to understand how critical syzygies can be used to generate
the complete set of total derivatives beyond the large-$\epsilon$ limit.
To study this question, we develop a computational approach to the construction of critical
syzygies. We then apply this approach to the cutting-edge two-loop $pp
\rightarrow t\overline{t}H$ process. In this way, we are able to demonstrate
that, in examples where the critical locus of the maximal-cut Baikov is
isolated, critical syzygies are indeed sufficient to generate all necessary
total derivatives.

The paper is organized as follows. In \cref{sec:preliminaries} we introduce our
setup of syzygies and surface terms in the Baikov representation. In
\cref{sec:theory} we discuss the syzygy method in the large-$\epsilon$ limit and
how this gives rise to the phenomenon of critical syzygies. In
\cref{sec:oneloop} we discuss the analytic construction of critical syzygies for
one-loop Feynman integrals and discuss where they do and do not generate a
complete set of total derivatives. In \cref{sec:twoloop}, we discuss critical
syzygies at two loops and explore their completeness computationally in a series
of examples focused on the two-loop five-point $pp \rightarrow t\overline{t}H$
process. Finally, in \cref{sec:summary}, we summarize and discuss future directions.

\section{Feynman Integral Relations in the Baikov Representation}
\label{sec:preliminaries}

In this section, we introduce the key objects under study: Feynman integrals and
their relations. We will organize these relations in the so-called ``syzygy''
formalism, particularly focusing on the construction of ``surface terms'', due
to their importance in the numerical unitarity method.
Let us consider a dimensionally-regulated, $l$-loop Feynman integral, that depends
on $E$ independent external momenta \(p_1, \ldots, p_E\).
Each such Feynman integral can be associated to a graph, $\Gamma$.
We will work in the Baikov representation~\cite{Baikov:1996rk}.
We refer the reader to ref.~\cite{Correia:2025yao} for a recent, detailed
derivation of the representation.
In the Baikov representation, a Feynman
integral $I_\Gamma(\mathcal{N}, \vec{\nu})$ associated to a graph $\Gamma$, with
numerator $\mathcal{N}$ and propagator powers $\vec{\nu}$ is
given as
\begin{equation}
    I_\Gamma(\mathcal{N}, \vec{\nu}) = \frac{c(D)}{G(p_1, \ldots, p_E)^{\gamma+l/2}} \int_{\mathcal{C}} \mathrm{d}^N \! \vec{z} \left[ B(\vec{z})^{\gamma}  \frac{\mathcal{N}}{\prod_{e \in \text{props}(\Gamma)} z_e^{\nu_e}} \right],
\label{eq:BaikovRepresentation}
\end{equation}
where $\gamma = (D - E - l - 1)/2$, $N = \binom{l}{2} + l E$ is the number of
Baikov variables (correspondingly the number of independent
scalar products in the diagram) and $c$ is an overall prefactor which depends
only on the dimensional regulator.
The function $G(p_1, \ldots, p_E)$ is the Gram determinant of the independent
external momenta in the Feynman integral, given by
\begin{equation}
    G(a_1, \ldots, a_m) = \det(a_i \cdot a_j).
\end{equation}
The denominator product in \cref{eq:BaikovRepresentation} is taken over the set of propagators associated to the
graph $\Gamma$, which we collect into the set of indices,
$\text{props}(\Gamma)$. We refer to the remaining set of Baikov variables as
``irreducible scalar products'', defining the associated set of indices
$\text{ISPs}(\Gamma)$ through
\begin{equation}
    \{ z_1, \ldots, z_N \} = \{z_e \,\, : \,\,  e \in \text{props}(\Gamma) \} \cup \{z_i \,\, : \,\, i \in \text{ISPs}(\Gamma) \}.
    \label{eq:BaikovVariablePartition}
\end{equation}
For notational convenience, as in \cref{eq:BaikovVariablePartition}, we will
often denote propagator or ``edge'' variables as $z_e$ and ISP variables as $z_i$.
The function $B(\vec{z})$ is known as the ``Baikov polynomial'' and can be
determined by expressing the
Gram determinant of the loop momenta and external momenta in the Feynman
integral in terms of Baikov variables.
The integration contour, $\mathcal{C}$ in \cref{eq:BaikovRepresentation}, has a
boundary given by the vanishing locus of the Baikov polynomial. That is,
\begin{equation}
  \partial \mathcal{C} = \{ \vec{z} \in \mathbb{R}^N \,\, : \,\, B(\vec{z}) = 0\}.
\end{equation}
More explicit details of the contour will not be needed for our discussion.

An important fact about dimensionally-regulated Feynman integrals is that, in
an appropriate representation, if the integrand is a total derivative, then the
integral is zero. In the Baikov representation, this arises as
\begin{equation}
0 = \int_{\mathcal{C}} \mathrm{d}^N\! \vec{z} \left[ \partial_k (B^{\gamma} f_k) \right],
\label{eq:TotalDerivative}
\end{equation}
where the $k$ index is summed over all Baikov variables and the $f_k$ are
rational functions of Baikov variables. A relation such as
\cref{eq:TotalDerivative} is known as an ``integration-by-parts'' or ``IBP''
relation.
IBP relations follow as the total derivative
integral can be re-written as an integral over the integration boundary,
$\partial \mathcal{C}$. However, as the the exponent $\gamma$ of the Bakov polynomial is taken
generic, the integrand vanishes on the boundary and the result is zero.
As the integrand of all Feynman
integrals comes with a factor of $B^\gamma$, it is useful to rewrite
\cref{eq:TotalDerivative} as
\begin{equation}
  0 = \int_{\mathcal{C}} \mathrm{d}^N \! \vec{z} \left[ B^\gamma \nabla_k f_k \right] ,
  \label{eq:TwistedTotalDerivative}
\end{equation}
where we introduce the twisted covariant derivative $\vec{\nabla}$, which acts as
\begin{equation}
   \nabla_k f_k = \partial_k f_k + \gamma f_k (\partial_k \log [B]).
   \label{eq:TwistedCovariantDerivative}
\end{equation}
If we consider appropriate choices of $\vec{f}$ in \cref{eq:TotalDerivative}
then we are led to the IBP relations for Feynman
integrals. The space of Feynman integrands modulo these IBP relations is known
as the space of master integrals.

A typical approach to the construction of relations between Feynman integrals is
to judiciously construct $\vec{f}$ and to act on them with $\vec{\nabla}$. The resulting
integrands then integrate to zero by \cref{eq:TwistedTotalDerivative}.
These relations can then be organized by linear algebra methods, broadly known
as the Laporta algorithm~\cite{Laporta:2000dsw}.
In ref.~\cite{Ita:2015tya}, it was observed that, when computing scattering
amplitudes, it can be important to be able to construct total derivatives with a
prescribed denominator structure.
The content of the integral relations then is
encoded in the numerators of the relations.
This leads to the definition of the vector space of so-called ``surface terms'' as
\begin{equation}
  \text{Surface}(\Gamma, \vec{\nu}) = \left\{ \, \mathcal{N} \,\,\, : \,\,\, \frac{\mathcal{N}}{\prod_{e \in \text{props}(\Gamma)} z_e^{\nu_e}} = \nabla_k\left[ \frac{a_k}{B^\Delta \prod_{e \in \text{props}(\Gamma)} z_e^{\beta_e}} \right] \, \right\},
  \label{eq:SurfaceSpaceDefinition}
\end{equation}
where $\Delta$ and the $\beta_e$ are non-negative integers and we require that
$\mathcal{N}$ and the $a_k$ belong to
\begin{equation}
  R = \mathbb{C}(p_{i} \cdot p_j, m_k^2, \epsilon)[z_1, \ldots, z_N].
  \label{eq:RingDefinition}
\end{equation}
That is, the $a_k$ are polynomials in Baikov variables, but rational in scalar
products of the external momenta, particle masses and the dimensional regulator $\epsilon$.
The set of surface terms is therefore an infinite dimensional $\mathbb{C}(p_i \cdot
p_j, m_k^2, \epsilon)$-subspace of the polynomial ring $R$. Understanding how to
explicitly construct a basis of the subspace of $\text{Surface}(\Gamma,
\vec{\nu})$ relevant for integral reduction is the main topic of this work.

In practice, controlling the integral relations via the Laporta algorithm, or
constructing surface terms can prove demanding. To this end, an observation made
in ref.~\cite{Gluza:2010ws} about Feynman integrals is that IBP relations can be
controlled by studying ``syzygies'', which allow one to directly generate
linear relations between Feynman integrals that do not have raised propagator
powers.
We will work with syzygies in the Baikov
representation, originally studied in ref.~\cite{Larsen:2015ped} and will consider syzygies of the form
\begin{equation}
   0 = a_0 B + \sum_{i \in \text{ISPs}(\Gamma)} a_i \partial_i B + \sum_{e \in \text{props}(\Gamma)} \tilde{a}_e z_e B + \sum_{e \in \text{props}(\Gamma)} \overline{a}_e z_e \partial_e B,
\label{eq:MasterSyzygy}
\end{equation}
where the \(a_0, a_i, \tilde{a}_e, \overline{a}_e\) are members of $R$.
That is, we are looking for tuples \(a_0, a_i, \tilde{a}_e, \overline{a}_e\) of
polynomials in Baikov variables, whose coefficients are rational functions in
the external kinematics and $\epsilon$, such that \cref{eq:MasterSyzygy} is
satisfied. We will denote the set of solutions to
\cref{eq:MasterSyzygy} as $\text{Syz}(\Gamma)$.
Importantly, the elements of $\text{Syz}(\Gamma)$ form an
$R-$module. That is, taking $R$-linear combinations of solutions of
\cref{eq:MasterSyzygy} yields other solutions.
We note that the syzygy relation \cref{eq:MasterSyzygy} is subtly different to
the one used in ref.~\cite{Larsen:2015ped}, due to the \(\tilde{a}_e\) term. In
practice, it has the effect that the \(a_0\) term can be studied on the maximal cut,
as all terms proportional to propagators can be moved into the \(\tilde{a}_e\).
This slight adjustment to the formalism of
ref.~\cite{Larsen:2015ped} will turn out to be fruitful later.

Let us consider how \cref{eq:MasterSyzygy} aids in the construction of integral
relations. We note that \cref{eq:MasterSyzygy} is a zero at the level of
polynomials, i.e. it is an integrand relation.
However, its form allows one to easily apply integration by parts in order to
end up with an interesting relation that is between integrals.
Specifically, by pre-multiplying \cref{eq:MasterSyzygy} with
$\frac{B^{\gamma-1}}{\prod_{e \in \text{props}(\Gamma)} z_e^{\nu_e}}$, we can rewrite it as
\begin{equation}
   0 = \frac{1}{\prod_{e \in \text{props}(\Gamma)} z_e^{\nu_e}} \left[a_0 B^{\gamma} + \frac{1}{\gamma} \sum_{i \in \text{ISPs}(\Gamma)} a_i \partial_i [B^\gamma] + \sum_{e \in \text{props}(\Gamma)} \left(\tilde{a}_e z_e B^\gamma + \frac{1}{\gamma} \overline{a}_e z_e \partial_e [B^\gamma] \right) \right].
\end{equation}
If we now integrate this $0$ over $\mathcal{C}$, we can perform partial
integrations on the $a_i$ and $\overline{a}_e$ terms to find a relation between
Feynman integrals,
\begin{equation}
  0 = \int_{\mathcal{C}} \mathrm{d}^N \! \vec{z} \left[ B^{\gamma} \frac{S_\Gamma(\vec{a}, \vec{\nu})}{\prod_{e \in \text{props}(\Gamma)} z_e^{\nu_e}}\right] ,
\label{eq:SurfaceTermDefinition}
\end{equation}
where we define 
\begin{equation}
S_{\Gamma}(\vec{a}, \vec{\nu}) = a_0 + \sum_{e \in \text{props}(\Gamma)} \tilde{a}_e z_e - \frac{1}{\gamma} \left[\sum_{i \in \text{ISPs}(\Gamma)} \partial_i a_i
    + \sum_{e \in \text{props}(\Gamma)} (z_e \partial_e \overline{a}_e - (\nu_e - 1)\overline{a}_e)  \right].
    \label{eq:CriticalSurfaceTerm}
\end{equation}
In this way, we see that a syzygy of the form \eqref{eq:MasterSyzygy} induces a
relation between Feynman integrals of the same dimension and without raising
propagator powers. It is clear that the $S_{\Gamma}(\vec{a}, \vec{\nu})$ of
eq.~\eqref{eq:SurfaceTermDefinition} is an element of $\text{Surface}(\Gamma,
\vec{\nu})$, and in fact $S_{\Gamma}$ maps the set of syzygyies to
the set of surface terms. This leads us to define
\begin{equation}
  \text{SyzSurface}(\Gamma, \vec{\nu}) = \left\{ S_{\Gamma}(\vec{a}, \vec{\nu}) \,\, : \,\, \vec{a} \in \text{Syz}(\Gamma) \right\},
\end{equation}
the set of surface terms constructed from syzygies, a manifest subspace of
$\text{Surface}(\Gamma, \vec{\nu})$. It is important to observe that there is no
claim that the space of surface terms arising from syzygies is the full space of
surface terms. Indeed, experience tells us that this is generally not the case.
One of the contributions of this work is to develop a criteria for when we can
expect the two spaces to be equal.

\section{Integral Relations and Critical Points}
\label{sec:theory}

When working in the syzygy formalism introduced in the previous section, one is
faced with the natural question of how to construct the set of syzygies,
$\text{Syz}(\Gamma)$. In practice, this turns out to be a difficult problem.
In this work, we make progress on this problem by using geometrical methods to
identify a subset of $\text{Syz}(\Gamma)$ which, in a broad set of cases,
generate a sufficient set of surface terms for reduction to master integrals.
To begin, we decide to consider the surface term in
\cref{eq:CriticalSurfaceTerm} for large values of the dimensional regularization
parameter, while also dropping terms that vanish on the maximal cut.
The perhaps surprising decision to consider the ``large-$\epsilon$ limit'' is
motivated by the success of this strategy in the context of intersection theory,
where taking this limit induces important simplifications~\cite{Mizera:2017rqa,
Mizera:2019vvs}.
In such a regime, we find that the surface term reduces to simply
\begin{equation}
	\left. \left[ \lim_{\epsilon \rightarrow \infty} S_\Gamma(\vec{a}, \vec{\nu}) \right] \right|_{\text{cut}_\Gamma} = a_0|_{\text{cut}_\Gamma},
  \label{eq:SurfaceTermLargeEpsilon}
\end{equation}
where by $f|_{\text{cut}_\Gamma}$, we mean that we evaluate $f$ on $z_e = 0$ for
$e \in \text{props}(\Gamma)$, i.e. we evaluate $f$ on the maximal cut of $\Gamma$.
We therefore see that, in this regime, we need to only consider the $a_0$ term
on the maximal cut.
This represents a
dramatic simplification of the IBP relation, as only a single term from the
syzygy relation contributes. This observation provides a strong motivation to
study the $a_0$ term of \cref{eq:MasterSyzygy} alone.
In this section, we shall study this piece geometrically and
interpret these features in the language of algebraic geometry.
For background on this language, we direct the reader to standard textbooks such
as~\cite{cox1994ideals}.

\subsection{Syzygies and Geometry}

Our aim is to understand the unknown polynomial $a_0$ in \cref{eq:MasterSyzygy}
by considering it geometrically. In principle, the techniques introduced
here can be used to study other terms, but we leave such investigations to further work.
In order to isolate the $a_0$ term in the syzygy,
it is natural to consider setting all of the other terms in
\cref{eq:MasterSyzygy} to zero. To this end, we consider setting
\begin{align}
  \begin{split}
    \partial_i B = 0, \qquad {i \in \text{ISPs}(\Gamma)}, \\
    z_e B = z_e \partial_e B = 0, \qquad {e \in \text{props}(\Gamma)}.
  \end{split}
  \label{eq:SyzygyVariety}
\end{align}
If we consider these equations as constraints on the $\vec{z}$ variables, we see
that, for fixed external kinematics, they cut out a surface in $\mathbb{C}^N$. As the equations in
\cref{eq:SyzygyVariety} are algebraic, this surface is an algebraic variety.
We will refer to this variety as the ``syzygy'' variety associated to $\Gamma$,
which we will denote as $U_{\text{syz}}^\Gamma$.
Importantly, for any point $\vec{z}$ on the syzygy variety, \cref{eq:MasterSyzygy} reduces to
\begin{equation}
  a_0 B |_{U_{\text{syz}}^\Gamma} = 0,
  \label{eq:a0GeometricalConstraint}
\end{equation}
and we see that we have successfully isolated the $a_0$ term in the syzygy. We
therefore see that $U_{\text{syz}}^\Gamma$ is of prime importance, so we
consider its structure.
Given that a number of the defining equations in \cref{eq:SyzygyVariety}
factorize, the syzygy variety can naturally be decomposed into subvarieties\footnote{In practical explorations, one also finds that there are
further decompositions that are not manifest in \cref{eq:SyzygyVariety}. We
leave systematic understanding of these branchings to further work.}.
A first observation is that it naturally splits into
two subvarieties where $B = 0$ and $B \ne 0$. That is,
\begin{equation}
   U_{\text{syz}}^\Gamma = U_{\text{sing}}^{\subseteq \Gamma} \cup U_{\text{crit}[\log(B)]}^\Gamma.
\end{equation}
The first variety, $U_{\text{sing}}^{\subseteq \Gamma}$, corresponds to the case $B=0$. 
By consideration of \cref{eq:SyzygyVariety} we see that it is composed of
a large number of subvarieties where either propagators are cut, or derivatives
of the Baikov polynomial are set to zero. That is, to each set of edges
$\Gamma_k \subseteq \Gamma$ we consider $U^{\Gamma_k}_{\text{sing}}$, the
variety in $\mathbb{C}^N$ defined by
\begin{align}
  \begin{split}
    \partial_i B &= 0 \,\,:\,\, {i \in \text{ISPs}(\Gamma_k)},
    \\
    z_e &= 0 \,\,:\,\, {e \in \text{props}(\Gamma_k)},
    \\
    B &= 0.
  \end{split}
  \label{eq:singularSubLocus}
\end{align}
Geometrically, each $U^{\Gamma_k}_{\text{sing}}$ corresponds to
the singular locus of the Baikov polynomial on the cut $\Gamma_k$, that is,
where the surface fails to be smooth. Explicitly, $U_{\text{sing}}^{\subseteq
  \Gamma}$ is the union of all these varieties, i.e.
\begin{equation}
  U^{\subseteq \Gamma}_{\text{sing}} = \bigcup_{\Gamma_k \subseteq \Gamma} U^{\Gamma_k}_{\text{sing}}.
\end{equation}
The second variety, $U_{\text{crit}[\log(B)]}^\Gamma$ corresponds to the case
where $B \ne 0$. 
Looking once again at
\cref{eq:SyzygyVariety}, we see that if the Baikov polynomial is non-zero, then the
propagators in $\Gamma$ must be zero and hence $U_{\text{crit}[\log(B)]}^\Gamma$
is defined by the equations
\begin{equation}
  \begin{split}
    \partial_i \log(B) &= 0 \,\,:\,\, {i \in \text{ISPs}(\Gamma)},
    \\
    z_e &= 0 \,\,:\,\, {e \in \text{props}(\Gamma)},
  \end{split}
  \label{eq:critLogBVariety}
\end{equation}
where we make use of the logarithm to enforce that $B \ne 0$. Geometrically,
we see that $U_{\text{crit}[\log(B)]}^\Gamma$ is the locus where the logarithm
of the $\Gamma$-cut Baikov polynomial reaches its extremal, or ``critical'' values.

Having understood the syzygy variety itself, let us consider what it tells us
about $a_0$. Considering \cref{eq:a0GeometricalConstraint}, in order to
further isolate the $a_0$ term, we impose that the factor of $B$ does not
vanish. That is, we consider \cref{eq:a0GeometricalConstraint} restricted
to $U_{\text{crit}[\log(B)]}^\Gamma$ and find that $a_0$ must vanish there, i.e.
\begin{equation}
  a_0|_{U_{\text{crit}[\log(B)]}^\Gamma} = 0.
  \label{eq:CriticalVarietyVanishing}
\end{equation}
In words, we see that $a_0$ vanishes on the critical locus of the logarithm of
the Baikov polynomial, on the cut corresponding to the graph $\Gamma$.
This observation is of strong importance, as this process-independent
constraint on a piece of the syzygy connects the syzygy formalism to other
recent advances in the understanding of relations between Feynman integrals.
Specifically, both the Lee-Pomeransky approach to counting master
integrals~\cite{Lee:2013hzt} as well as recent advances in the application of
intersection theory to Feynman integrals~\cite{Mizera:2019vvs} make use of the
variety ${U_{\text{crit}[\log(B)]}^\Gamma}$.
It is well-known that in many cases this variety is a finite set of points and
Lee and Pomeransky showed in ref.~\cite{Lee:2013hzt} that, in such cases, the
number of these points, when counted with multiplicity, is exactly the number of
master integrals associated to the topology $\Gamma$.
Moreover, this statement is reinforced in the intersection
theory literature, where intersection numbers can be written in terms of
evaluations of the integrand on the points of
${U_{\text{crit}[\log(B)]}^\Gamma}$. It is therefore perhaps not surprising that
syzygies of Feynman integrals would also exhibit a connection to
${U_{\text{crit}[\log(B)]}^\Gamma}$.
In this work, we shall explore how to constructively use this connection to build
syzygies for Feynman integrals.

\subsection{From Geometry to Algebra}
\label{sec:GeometryToAlgebra}

In the previous subsection, we have gained a geometric insight into a piece of
the syzygy relation~\eqref{eq:MasterSyzygy}, demonstrating a connection to
the variety $U_{\text{crit}[\log(B)]}^\Gamma$.
However, the connection currently remains unconstructive. While
\cref{eq:CriticalVarietyVanishing} tells us that all $a_0$ must vanish on the
critical locus of $\log(B)$ on the cut associated to $\Gamma$, it is unclear if
this property is sufficient or simply necessary. Indeed, when considering
\cref{eq:a0GeometricalConstraint} closely, we notice a potential subtlety:
we are unable to exclude the possibility that $a_0$ must also vanish on
$U_{\text{sing}}^{\subseteq \Gamma}$.
In order to gain control of this subtlety, we shall
rephrase our geometric discussion in the
algebraic language of ideals.
Let us consider the terms of
\cref{eq:MasterSyzygy} other than $a_0B$. By inspection, we see that they
parameterize an element of the ideal
\begin{equation}
  J_{\text{syz}}^\Gamma = \langle  \partial_i B \,\, : \,\, i \in \text{ISPs}(\Gamma)\rangle + \langle z_e B, z_e \partial_e B \,\, : \,\, e \in \text{props}(\Gamma) \rangle.
  \label{eq:JSyzDefinition}
\end{equation}
Here we denote the ideal generated by $\{g_1, \ldots \}$ as $\langle  g_1,
\ldots \rangle$ and $+$ denotes the ideal sum.
The ideal $J_{\text{syz}}^\Gamma$ is an ideal of the polynomial ring $R$ defined
in \cref{eq:RingDefinition}: polynomials in Baikov variables, with coefficients
that are rational functions of external kinematics and $\epsilon$.
Importantly, the syzygy variety that we identified earlier,
$U_{\text{syz}}^\Gamma$ is the variety associated to the ideal
$J_{\text{syz}}^\Gamma$. That is\footnote{ We recall that the variety $V(J)$
associated to an ideal $J$ of $\mathbb{F}[z_1, \ldots, z_N]$, for some field
$\mathbb{F}$, is the set of $\vec{z} \in \mathbb{F}^N$ such that $p(\vec{z}) =
0$ for all $p \in J$ and refer the reader to ref.~\cite{cox1994ideals} for more
details.},
\begin{equation}
  U_{\text{syz}}^\Gamma = V(J_{\text{syz}}^\Gamma).
\end{equation}
The importance of $J_{\text{syz}}^\Gamma$ is that it encodes algebraic features
of the syzygies, such as multiplicity, which we can associate to the variety
$U_{\text{syz}}^\Gamma$.
Having seen the importance of $U_{\text{syz}}^\Gamma$, let us similarly analyze
$J_{\text{syz}}^\Gamma$.
Similar to the splitting of the associated variety, it
is possible to prove an analogous splitting for $J_{\text{syz}}^\Gamma$.
In contrast to splitting a variety, which expresses it as the union of
multiple subvarieties, an ideal can be expressed as the intersection of
other ideals, each of which is larger than the initial ideal. A particularly
relevant splitting of $J_{\text{syz}}^\Gamma$ is induced by the explicit factors
of $B$ in some of its generators. In order to perform this splitting, we make use of
the lemma proven in appendix \ref{app:splittingLemma} and write $J_{\text{syz}}^\Gamma$ as
\begin{equation}
  J_{\text{syz}}^\Gamma = J_{\text{sing}}^{\subseteq \Gamma, \mu} \cap
  J_{\text{crit}(B)}^\Gamma,
  \label{eq:JsyzSplit}
\end{equation}
where we define
\begin{align}
  \label{eq:JSingDef}
  J_{\text{sing}}^{\subseteq \Gamma, \mu} &= J_{\text{syz}}^\Gamma + \langle B^{\mu} \rangle, \\
  J_{\text{crit}(B)}^\Gamma &= \langle \partial_i B : i \in \text{ISPs}(\Gamma) \rangle + \langle z_e : e \in \text{props}(\Gamma)\rangle.
\end{align}
The two ideals get their names from their
associated varieties as $V(J_{\text{sing}}^{\subseteq \Gamma, \mu}) =
U_{\text{sing}}^{\subseteq \Gamma}$ and $V(J_{\text{crit}(B)}^\Gamma)$ is the
critical locus of the Baikov polynomial on the cut associated to $\Gamma$.
The exponent $\mu$ in \cref{eq:JSingDef} is a positive integer known as the ``saturation
index'' of $J_{\text{syz}}^\Gamma$ with respect to the Baikov polynomial.
It represents the multiplicity of the $B=0$ component of
$J_{\text{syz}}^\Gamma$, and its value is {\em a priori} unknown. Practical experience
says that it is often $1$, but there exist physical examples, such as that
discussed in \cref{sec:tthApplication}, where it is higher.

Having introduced $J_{\text{syz}}^\Gamma$ let us now consider how we can use it
to understand the $a_0$ term. From \cref{eq:MasterSyzygy} we see that $a_0$
is any polynomial in $R$, such that when you multiply it by the Baikov
polynomial, you get an element of $J_{\text{syz}}^\Gamma$. This leads us to
define the set of all possible $a_0$ terms as
\begin{equation}
  A_0^\Gamma =  \left\{ \, p \in R \,\,\, : \,\,\, p B \, \in \, J_{\text{syz}}^\Gamma \, \right\}.
  \label{eq:A0Definition}
\end{equation}
Importantly, it is not hard to see that $A_0^\Gamma$ is also an ideal of $R$.
By consideration of \cref{eq:SurfaceTermLargeEpsilon} we are therefore able
to make the remarkable statement that on the maximal cut and in the large-$\epsilon$ limit, surface terms actually have the structure of an ideal.
This observation will allow us to better understand the structure of surface
terms using methods from the theory of ideals.
In general, explicitly finding a generating set for $A_0^\Gamma$ is a
non-trivial task. For the moment, we content ourselves with making structural
statements.

We begin by observing that the set in \cref{eq:A0Definition} is actually the
definition of the ``ideal quotient'' of $J_{\text{syz}}^\Gamma$ by the ideal
generated by the Baikov polynomial.
That is, we have that
\begin{equation}
  A_0^\Gamma = J_{\text{syz}}^\Gamma : \langle B \rangle.
  \label{eq:A0AsQuotient}
\end{equation}
The relation in \cref{eq:A0AsQuotient} is the crucial constructive observation
of this work, allowing us to study the problem of determining syzygies with the
technology of ideal quotients.
A first important property of ideal quotients is that they have a geometrical
significance. Specifically, the ideal quotient can be used to
implement the set difference of two varieties: to remove one variety
from another. The relation relevant to our construction is
\begin{equation}
  V[J_{\text{syz}}^\Gamma : \langle  B^\mu \rangle] = \overline{V[J_{\text{syz}}^\Gamma] \setminus V(\langle B \rangle)} = U_{\text{crit}[\log(B)]}^\Gamma,
  \label{eq:SaturationSetDifference}
\end{equation}
where the bar represents that we take the Zariski closure\footnote{The Zariski
  closure of a set is the smallest algebraic variety that contains that set. In
  the context of ideal quotients, this has the effect of filling in ``holes'' in the
  variety. Further details will not be required for our discussion.}.
That is, if one quotients $J_{\text{syz}}^\Gamma$ by $B^\mu$, the associated
variety is $U_{\text{syz}}^\Gamma$ with the $B=0$ component removed. From the
previous section, we see that this is just $U_{\text{crit}[\log(B)]}^\Gamma$.
Importantly,
\cref{eq:SaturationSetDifference} gives us a geometrical interpretation of the
saturation index $\mu$.
As $\mu$ represents the multiplicity of the $B = 0$ component of $J_{\text{syz}}^\Gamma$,
removing it requires quotienting $J_{\text{syz}}^\Gamma$ by $B$ ``$\mu$ times''.
A second important property of ideal quotients is that they act on each
component of an intersection\footnote{That is, for $R$-ideals $A$, $B$ and $C$,
one has that $(A \cap B) : C = (A : C) \cap (B : C)$.}.
By consideration of \cref{eq:JsyzSplit}, this allows us to write $A_0^\Gamma$ as an
intersection of ideals as
\begin{equation}
  A_0^\Gamma = \left(J_{\text{crit}(B)}^\Gamma : \langle B \rangle\right) \cap \left( J_{\text{sing}}^{\subseteq \Gamma, \mu} : \langle B \rangle \right).
  \label{eq:StrictIdealQuotient}
\end{equation}
This relation allows us to robustly understand the connection of $A_0^\Gamma$ to our
geometrical considerations. First, let us consider the set of all polynomials
that vanish on $U_{\text{crit}[\log(B)]}^\Gamma$, denoted
$I(U^\Gamma_{\text{crit}[\log(B)]})$. By applying Hilbert's strong Nullstellensatz to
\cref{eq:SaturationSetDifference}, we have that
\begin{equation}
   I(U^\Gamma_{\text{crit}[\log(B)]}) = \sqrt{J_{\text{crit}(B)}^\Gamma : \langle B^\mu \rangle},
\end{equation}
where the square root of an ideal denotes taking its radical. Using elementary
properties of radicals, quotients and intersections, it is not difficult to
conclude that
\begin{equation}
  A_0^\Gamma \subseteq I(U^\Gamma_{\text{crit}[\log(B)]}).
\end{equation}
That is, in general, $A_0^\Gamma$ is a subset of the polynomials which vanish on
$U^\Gamma_{\text{crit}[\log(B)]}$.
In other words, we see that while $a_0$ must vanish on
$U_{\text{crit}[\log(B)]}^\Gamma$, it may also satisfy further non-trivial
constraints.

In practice, it turns out that there is an important case where
\cref{eq:StrictIdealQuotient} can be shown to simplify: where
$\mu = 1$. As we will discuss in \cref{sec:oneloop,sec:twoloop}, we
experimentally find that this is almost always the case for Feynman integrals.
Let us, therefore, analyze the $\mu = 1$ case in detail.
In this case, if we look to the definition of $J_{\text{sing}}^{\subseteq
\Gamma, \mu}$ in \cref{eq:JSingDef}, we see that any polynomial in $R$
multiplied by the Baikov polynomial is in $J_{\text{sing}}^{\subseteq \Gamma,
\mu}$, as the Baikov polynomial is a generator.
Hence, we see that the right intersectand of \cref{eq:StrictIdealQuotient} is
$R$ and as $R$ is the identity under intersection, we conclude that
\begin{equation}
\mu = 1 \qquad \Rightarrow \qquad A_0^\Gamma = J^\Gamma_{\text{crit}(B)} : \langle B \rangle.
\label{eq:SimplifiedIdealQuotient}
\end{equation}
That is, our ideal quotient simplifies exactly if $\mu = 1$.
By the definition of the saturation index, we see that, in this case, all $a_0$
terms must vanish on $U_{\text{crit}[\log(B)]}^\Gamma$.
An important question that we will study experimentally in this work is what
value of $\mu$ we will typically encounter. In practice, we will find that it is
most often $1$. Nevertheless, we will return to the question of how to interpret
cases where $\mu \ne 1$ in \cref{sec:NonTrivialSaturation}.

To close this section, let us note that there is a simple case where we can
easily find a generating set for $A_0^\Gamma$. It follows by constructing a simple
upper and lower bound for $A_0^\Gamma$. Specifically, one has that
    \begin{equation}
	    J_{\text{crit}(B)}^\Gamma \subseteq A_0^\Gamma \subseteq J_{\text{crit}(B)}^\Gamma :\langle B \rangle.
      \label{eq:A0Inclusion}
    \end{equation}
The left inclusion follows by inspection of the generators of
$J_{\text{crit}(B)}^\Gamma$. Each of them, when multiplied by the Baikov
polynomial, is an element of $J_{\text{syz}}^\Gamma$, and hence the left
inclusion follows. The right inclusion follows by consideration of
\cref{eq:StrictIdealQuotient}.
Now, let us consider the case that $J_{\text{crit}(B)}^\Gamma =
J_{\text{crit}(B)}^\Gamma : \langle B \rangle$. Then we have that the upper and
lower bounds of \cref{eq:A0Inclusion} are equal leading to $A_0^\Gamma =
J_{\text{crit}(B)}^\Gamma$. In such a case, as a generating set of
$J_{\text{crit}(B)}^\Gamma$ is given, then we find a generating set for $A_0^\Gamma$.
Interpreting this geometrically, we see that this corresponds to the case where
$U^{\Gamma}_{\text{sing}}$ is empty, i.e. the zero-set of the maximal cut Baikov
polynomial is a smooth variety.

\subsection{Syzygies and Critical Points}
\label{sec:SyzygySurfaceTermConnection}

Having understood the $a_0$ term in the syzygy relation
\cref{eq:MasterSyzygy} both geometrically and algebraically, we will now argue
that the $a_0$ term is of deep importance to integration-by-parts relations.
To this end, we recall the approach of Lee and Pomeransky in
ref.~\cite{Lee:2013hzt} for counting the number of master integrals associated
to a Feynman integral with graph $\Gamma$.
To begin, let us identify the space of master integrals on the maximal cut of
$\Gamma$ as the linearly independent numerators modulo surface terms and terms
that vanish on the cut. That is,
\begin{equation}
  H_\Gamma = R / (\text{Surface}(\Gamma, \vec{1}) + J_{\text{cut}}^\Gamma),
\end{equation}
where we make use of
\begin{equation}
J_{\text{cut}}^\Gamma = \langle  z_e : e \in \text{props}(\Gamma) \rangle
\end{equation}
and by $\vec{1}$ we mean that the entries of the exponent vector $\vec{\nu}$
are all 1.
The question of counting the number of master integrals can then be understood
as counting the dimension of the vector space $H_\Gamma$.
The approach of Lee and Pomeransky argues that $\mathrm{dim}(H_\Gamma)$ is encoded in the solution set of the equations
\begin{align}
  \begin{split}
    \partial_i B &= 0 \,\, : \,\, i \in \text{ISPs}(\Gamma),
    \\
    z_e  &= 0 \,\, : \,\, e \in \text{props}(\Gamma),
    \\
    B &\ne 0.
  \end{split}
  \label{eq:LeePomeranskyEquations}
\end{align}
Precisely, in the case that the solution set of \cref{eq:LeePomeranskyEquations}
is a finite number of points, then the number of master integrals,
$\mathrm{dim}(H_\Gamma)$, is equal to the number of solutions, counted with
multiplicity. As alluded to earlier, note that \cref{eq:LeePomeranskyEquations}
is entirely equivalent to \cref{eq:critLogBVariety} and so the Lee and
Pomeransky approach is counting points in
$U_{\text{crit}[\log(B)]}^\Gamma$.

In ref.~\cite{Lee:2013hzt}, in order to count the number of solutions to
\cref{eq:LeePomeranskyEquations} without directly computing the set of
solutions, the authors introduce an ``algebraic formulation'', which is
implemented in the code \texttt{MINT}. This formulation makes use of the ideal
\begin{equation}
  J_{\text{LP}}^{\Gamma}  = \langle \partial_i B \,\, : \,\, i \in \text{ISPs}(\Gamma), \quad z_e  \,\, : \,\, e \in \text{props}(\Gamma), 1 - w B \rangle_{R[w]} \cap R,
  \label{eq:RabinowitschIdeal}
\end{equation}
where $w$ is an auxiliary variable.
The explicit ideal on the right-hand-side of \cref{eq:RabinowitschIdeal} is an
ideal in the ring $R[w]$, the ring of both Baikov variables and $w$. The
intersection with the ring $R$ eliminates the variable $w$, and can be
implemented with Groebner basis techniques.
The ideal $J_{\text{LP}}^\Gamma$ can be used to compute the number of points in
$V(J^{\Gamma}_{\text{LP}})$ by exploiting special properties of an ideal $J$,
whose associated variety $V(J)$ is a finite set of points.
Specifically, the number of points, in $V(J)$ counted with multiplicity, is equal
to the number of linearly independent polynomials when they are considered
modulo the ideal $J$ \cite[Chapter 4, Corollary 2.6]{cox2005using}.
That is, one takes the ring of polynomials $R$ and considers any two elements
or $R$ that differ by an element of $J$ to be equivalent.
Denoting this equivalence class ring as $R/J$, in the case where $V(J)$ is a
finite set of points, it turns out that $R/J$ is a finite-dimensional vector
space. The dimension of this vector space, $\text{vdim}(R/J)$ gives the desired
point counting.
Applying this to the ideal $J_{\text{LP}}^{\Gamma}$ we arrive at the
Lee-Pomeransky formula for counting master integrals,
\begin{equation}
\dim(H_\Gamma) = \text{vdim}\left(R / J_{\text{LP}}^{\Gamma}\right).
\label{eq:LeePomeranskyCriterion}
\end{equation}
Importantly, $\text{vdim}\left(R / J_{\text{LP}}^{\Gamma}\right)$ can be computed without
computing the solutions to \cref{eq:LeePomeranskyEquations}. It requires only a
Groebner basis of $J_{\text{LP}}^{\Gamma}$, which can easily be computed in modern
computer algebra systems.

Naturally, the master integral counting of Lee and Pomeransky must in some way
be connected to the set of syzygies, as they are intimately related to the
construction of the surface terms. Nevertheless, it turns out that this
connection can be made much more directly.
Let us return to consider \cref{eq:RabinowitschIdeal}. An important observation
is that the ideal in \cref{eq:RabinowitschIdeal} can be identified as an
application of the so-called ``Rabinowitsch trick'' to perform the ideal
saturation of $J_{\text{crit}(B)}^\Gamma$ with respect to $B$ (see
ref.~\cite[Chapter 4, Section 4, Theorem 14 (ii)]{cox1994ideals}).
That is, we have that
\begin{equation}
  J_{\text{LP}}^\Gamma = J_{\text{crit}(B)}^\Gamma : \langle  B^\mu \rangle,
  \label{eq:RabinowitschTrick}
\end{equation}
where $\mu$ is the saturation index of $J_{\text{crit}(B)}^\Gamma$ with respect
to the Baikov polynomial.
Therefore, in the case $\mu = 1$, we have that the Lee-Pomeransky ideal is
exactly the ideal to which all $a_0$ terms must belong. That is
\begin{equation}
  \mu = 1 \qquad \Rightarrow \qquad J_{\text{LP}}^\Gamma = A_0^\Gamma
\end{equation}
This is a striking statement: the ideal involved in counting master integrals
with the Lee-Pomeransky approach also arises in the syzygy approach.

The observation of the direct connection between the syzygy approach and 
the Lee-Pomeransky approach has important consequences for considering total
derivatives.
To see this, consider that $H_\Gamma$ and $R/J_{\rm LP}^\Gamma$ are two different quotient
spaces of $R$ and as such they each furnish a decomposition of $R$ as a vector
space.
We can therefore write
\begin{equation}
    \sigma_1(H_\Gamma) + \text{Surface}(\Gamma, \vec{1}) + J_{\text{cut}}^\Gamma = \sigma_2(R/J_{\text{LP}}^\Gamma) + J_{\text{LP}}^\Gamma,
    \label{eq:RDecompEquality}
\end{equation}
where the $\sigma_i$ are canonical maps that identify the quotient spaces as
subspaces of $R$. That is, both the left- and right-hand side of
\cref{eq:RDecompEquality} can be recognized as a decomposition of $R$ into
$R/W+W$ for some subspace $W$ and are therefore equal.
Let us now assume that there exists a set of master integrals
that are linearly independent on the points of $V(J_{\text{LP}}^\Gamma)$, i.e.
on $U_{\text{crit}[\log(B)]}$.
This is the statement that we can choose the $\sigma_i$ such that
$\sigma_1(H_{\Gamma}) = \sigma_2(R/J_{\text{LP}}^\Gamma)$. In experience of
practical applications to Feynman integrals, this is found to be true so we
assume this from now on.
Under this assumption, it is a simple application of a standard fact of linear
algebra that the two spaces complementary to the $\sigma_i$ spaces in
\cref{eq:RDecompEquality} are isomorphic. Under the condition that $\mu=1$ and
$U^\Gamma_{\text{crit}[\log(B)]}$ is a set of points, we are therefore able to conclude that
\begin{equation}
	A_0^\Gamma \simeq \text{Surface}(\Gamma,\vec{1}) + J_{\text{cut}}^\Gamma.
    \label{eq:LPSurfaceCorrespondence}
\end{equation}
This equation, \cref{eq:LPSurfaceCorrespondence}, is an important result of our
work. We interpret \cref{eq:LPSurfaceCorrespondence} to say that a basis of
$A_0^\Gamma$ is in one-to-one correspondence with a basis of surface terms, up to terms
that vanish on the cut associated to $\Gamma$.

\subsection{Critical Surface Terms}
\label{sec:CriticalSurfaceTerms}

The correspondence between the space of $a_0$ parts of syzygies and surface
terms that we have just identified suggests an interesting perspective on the
syzygy formalism for surface term construction.
Specifically, under the conditions that \cref{eq:LPSurfaceCorrespondence} holds,
we should only need a set of syzygies whose $a_0$ piece is linearly independent
on the maximal cut, in order to be able to reduce to master integrals, modulo
pinch integrals.
This is a highly compelling observation as there are a large number of elements
of $\text{Syz}(\Gamma)$ that satisfy $a_0 = 0$, but
\cref{eq:LPSurfaceCorrespondence} tells us that we can neglect them.
Motivated by this, we will now define the set of ``critical surface terms'': those
with linearly independent $a_0$ piece on the maximal cut.

To begin, let us define the ``critical part'' of an element $\vec{a}$ of
$\text{Syz}(\Gamma)$ as the $a_0$ part, i.e.
\begin{equation}
    \mathfrak{c}(\vec{a}) \coloneq{} a_0.
    \label{eq:CriticalPartDefinition}
\end{equation}
Our discussion states that two syzygies that have the same critical part give
the same on-shell relation between Feynman integrals in the large-$\epsilon$
limit. Beyond this, the introduction of the $\tilde{a}_e$ in
\cref{eq:MasterSyzygy} terms means that two syzygies will give rise to the same
surface term if they differ by an $R$-linear combination of syzygies of the form
\begin{equation}
  a_0 = z_e, \qquad \tilde{a}_e = -1.
  \label{eq:TrivialSyzDef}
\end{equation}
with all other entries being 0. We will denote the submodule of
$\text{Syz}(\Gamma)$ generated by the syzygies of \cref{eq:TrivialSyzDef} as
$\text{ZSyz}(\Gamma)$ as they give rise to surface terms which are zero.
Together, this leads us to observe that there is a natural equivalence relation
on $\text{Syz}(\Gamma)$: two elements that differ either by some $\vec{w}_1 \in
\text{ZSyz}(\Gamma)$ or $\vec{w}_2 \in \text{Syz}(\Gamma)$ such that
$\mathfrak{c}(\vec{w}_2) = 0$ should be regarded as equivalent.
Denoting this equivalence relation as $\sim$, we have
\begin{equation}
  \vec{a} \,\, \sim \,\, \vec{a} + \vec{w}_1 + \vec{w}_2, \qquad \text{where} \quad \vec{a} \in \text{Syz}(\Gamma), \quad \vec{w}_1 \in \text{ZSyz}(\Gamma) \quad \text{and} \quad \vec{w}_2 \in \ker(\mathfrak{c}).
\end{equation}
This motivates us to define the module of ``critical syzygies'' as the quotient
of the module of syzygies by this equivalence relation. That is, we define
\begin{equation}
  \text{CSyz}(\Gamma) = \text{Syz}(\Gamma) / (\ker(\mathfrak{c}) + \text{ZSyz}(\Gamma)),
  \label{eq:CSyzDefinition}
\end{equation}
where we recognize the set of elements that are equivalent to zero under $\sim$
as the module sum of the kernel of $\mathfrak{c}$ and $\text{ZSyz}(\Gamma)$.
Elements of $\text{CSyz}(\Gamma)$ are equivalence classes under $\sim$, and can
be represented by elements of $\text{Syz}(\Gamma)$. For any element $\vec{a}$ of
$\text{Syz}(\Gamma)$, we denote the associated equivalence class in
$\text{CSyz}(\Gamma)$ as $[\vec{a}]$. Conversely, given $[\vec{a}] \in \text{CSyz}(\Gamma)$,
we will refer to a (non-unique) representative $\vec{a} \in \text{Syz}(\Gamma)$ as a
lift of $[\vec{a}]$.
An important observation about $\text{CSyz}(\Gamma)$ is its module structure. By
construction, two elements of $\text{CSyz}(\Gamma)$ are inequivalent only if
their critical parts on the maximal cut are distinct. We therefore see that
\begin{equation}
\text{CSyz}(\Gamma) \simeq A_0^\Gamma/J_{\text{cut}}^\Gamma.
\label{eq:CriticalSyzModuleStructure}
\end{equation}
It is therefore clear that the set of critical syzygies is much ``smaller'' than
the full set of syzygies, as it is only a rank $1$ module.

Let us now consider using $\text{CSyz}(\Gamma)$ to construct surface terms.
A technicality here is that elements of $\text{CSyz}(\Gamma)$ are equivalence
classes of syzygies. That is, $\text{CSyz}(\Gamma)$ is a quotient space of
$\text{Syz}(\Gamma)$. Nevertheless, as a quotient space of $\text{Syz}(\Gamma)$,
$\text{CSyz}(\Gamma)$ is isomorphic to some subspace of $\text{Syz}(\Gamma)$.
That is, analogous to the $\sigma_i$ of \cref{eq:RDecompEquality}, there exists
a linear map $\pi$ such that
\begin{equation}
  \pi(\text{CSyz}(\Gamma)) \subset \text{Syz}(\Gamma).
\end{equation}
Practically, defining $\pi$ can be thought of as finding a lift $\vec{a}_i$ for each basis
vector $[\vec{a}_i]$ of $\text{CSyz}(\Gamma)$. Naturally, the non-uniqueness of this
lift implies that $\pi$ is not unique.
Nevertheless, this map $\pi$ allows us to define a subspace of surface terms,
built from critical syzygies as\footnote{We note that the
  space $\text{CSyzSurface}(\Gamma, \vec{\nu})$ depends on the choice of $\pi$.
  However, under the conditions of $U_{\text{crit}[\log(B)]}^\Gamma$ being points and $\mu = 1$, the
  space depends only on $\pi$ through surface terms which vanish on the maximal
  cut. We therefore choose to suppress $\pi$ in the notation.}
\begin{equation}
    \text{CSyzSurface}(\Gamma, \vec{\nu}) = \{ S_{\Gamma}(\vec{a}, \vec{\nu}) : \vec{a} \in \pi(\text{CSyz}_{\Gamma}) \}.
    \label{eq:CSyzSurfaceDefinition}
\end{equation}
Interestingly, this construction gives a new perspective on the syzygy equation,
\cref{eq:MasterSyzygy}.
In the case where $\mu = 1$ and $U_{\text{crit}[\log(B)]}^\Gamma$ is isolated
then it can be seen as a recipe to lift elements of the Lee-Pomeransky ideal
$J_{\text{LP}}^\Gamma$ to surface terms.

The importance of the set of critical surface terms arises as we have
constructed it to be a set of surface terms that is sufficient in the 
large-$\epsilon$, maximal-cut limit and therefore can be used as a tool to fill
the space of surface terms and perform a reduction to master integrals.
To see this explicitly, consider that $S_{\Gamma}(\vec{a})$ reduces to $\mathfrak{c}(\vec{a})$
in the large-$\epsilon$ limit.
Looking to \cref{eq:CSyzSurfaceDefinition}, we therefore see that the set of
critical surface terms becomes the set of maximal-cut $a_0$ terms of
$\text{CSyz}(\Gamma)$ in this limit.
As this set of terms is $A_0^\Gamma/J^{\Gamma}_{\text{cut}}$ by
\cref{eq:CriticalSyzModuleStructure}, we therefore see that
$\text{CSyzSurface}(\Gamma, \vec{\nu})$ is isomorphic to
$A_0^\Gamma/J_{\text{cut}}^\Gamma$.
If we now consider \cref{eq:LPSurfaceCorrespondence}, we see that, given $\mu =
1$ and $U_{\text{crit}[\log(B)]}^\Gamma$ a finite collection of points, we have
that
\begin{equation}
   \text{CSyzSurface}(\Gamma, \vec{\nu}) \simeq \text{Surface}(\Gamma, \vec{1})/(\text{Surface}(\Gamma, \vec{1}) \cap J_{\text{cut}}^\Gamma),
   \label{eq:CriticalSurfaceSufficiency}
\end{equation}
where we have made use of the fact that, for two vector spaces $V$ and $W$ that
are subspaces of some larger space, $(V+W)/W \simeq V/(V\cap W)$.
Restricting now to the $\vec{\nu}=\vec{1}$ case,
we see that by \cref{eq:CSyzSurfaceDefinition} and
\cref{eq:CriticalSurfaceSufficiency} we have that $\text{CSyzSurface}(\Gamma,
\vec{1})$ is a subspace of $\text{Surface}(\Gamma, \vec{1})$ that is isomorphic
to $\text{Surface}(\Gamma, \vec{1})$, modulo surface terms that vanish on the
maximal cut.
Naturally, these surface terms that vanish on the cut can be captured by
repeating the procedure for pinch topologies.
Therefore, by iteratively constructing the set of all
$\text{CSyzSurface}(\Gamma_k, \vec{1})$ for all $\Gamma_k \subseteq \Gamma$,
one can construct the full set of surface terms.

\subsection{Multiplicity Structure}
\label{sec:NonTrivialSaturation}

To close out our discussion of syzygies and critical surface terms, we return to
interpreting the case where the saturation index $\mu \ne 1$, that is, where
the multiplicity of the $B=0$ component of the syzygy ideal
$J_{\text{syz}}^\Gamma$ is not $1$.
In the analysis of \cref{sec:SyzygySurfaceTermConnection}, this stopped us from
being able to conclude that critical syzygies provide all relations among
Feynman integrals.
In the following, we argue that by expanding our construction, we can find
further relations.
To this end, we consider an alternative construction of the syzygy formalism.
Rather than starting directly from a judiciously chosen syzygy equation such as
\cref{eq:MasterSyzygy}, we make contact with the definition of surface terms in
\cref{eq:SurfaceSpaceDefinition}.
The driving observation is that when we construct total derivatives from the
syzygy of \cref{eq:MasterSyzygy} we
implicitly restrict the structure of our total derivatives. Specifically, we
assume that the exponent of the Baikov polynomial inside the total derivative is the same
as that of the integrals that are targeted for reduction. While the constraint
is natural, releasing this constraint can allow for more integral relations, as
we will see.

We begin by considering numerator polynomials $\mathcal{N}$ that arise from
total derivatives as
\begin{equation}
  \frac{\mathcal{N}B^\gamma}{\prod_{e\in {\mathrm{props}(\Gamma)}} z_e}
  =
  \sum_{i \in {\rm ISPs}(\Gamma)} \partial_i \left[ \frac{B^{\gamma - \Delta} a_i}{\prod_{e\in {\mathrm{props}(\Gamma)}} z_e}\right]
  +
  \sum_{e' \in {\rm props}(\Gamma)} \partial_{e'} \left[ \frac{B^{\gamma - \Delta} z_{e'} \overline{a}_{e'}}{\prod_{e\in {\mathrm{props}(\Gamma)}} z_e} \right].
    \label{eq:DimShiftSeeding}
\end{equation}
Here, we introduce directly the propagator non-doubling constraint by including
a factor of $z_{e'}$ in the numerator of the second term.
Note that here we have made use of unit propagator powers for ease of analysis.
Importantly, in \cref{eq:DimShiftSeeding}, we have introduced a (non-negative)
integer $\Delta$.
If $\Delta = 0$, then this corresponds to the syzygy analysis
in \cref{sec:preliminaries}, while $\Delta > 0$ is more general.
From the perspective of the Laporta algorithm, \cref{eq:DimShiftSeeding} makes
use of ``seed'' integrals (those of which we take the derivative) with a power
of the Baikov polynomial that is lower than that of the target integrals.
Recalling that $\gamma = (D - E - l - 1)/2$, one effectively considers seed
integrals defined in $D - 2 \Delta$ dimensions, rather than $D$ dimensions.
Let us expand out the argument of the total derivative and remove the common
factor of $\frac{B^{\gamma - \Delta - 1}}{\prod_{e\in {\mathrm{props}(\Gamma)}} z_e}$, leading to 
\begin{equation}
	\mathcal{N} B^{\Delta+1} =  (\gamma - \Delta) \left[\sum_{i\in {\rm ISPs}(\Gamma)} \! a_i\partial_i B + \!\!\sum_{e \in {\rm props}(\Gamma)} \!\overline{a}_e z_e \partial_e B \right]
	+ B \! \left[ \sum_{i\in {\rm ISPs}(\Gamma)} \!\partial_i a_i + \!\!\sum_{e \in {\rm props}(\Gamma)} \!z_e \partial_e \overline{a}_e \right]\!.
  \label{eq:AugmentedSyzygy}
\end{equation}
Next, we must impose that the right-hand side of \cref{eq:AugmentedSyzygy}
is proportional to $B^{\Delta + 1}$. We see that this takes on a different
character, depending on the value of $\Delta$.
If $\Delta = 0$, then the second term on the right-hand side of
\cref{eq:AugmentedSyzygy} is already proportional to $B$ and so one only has the
constraint that the first term is proportional to the Baikov polynomial. This leads to the
style of syzygy in \cref{eq:MasterSyzygy}.

Let us consider making the right-hand side of \cref{eq:AugmentedSyzygy}
proportional to $B^{\Delta+1}$ for general $\Delta$.
We see that the second term can no longer be ignored and that the
proportionality constraints now involve not only syzygy-like terms, but also
terms involving derivatives of the $a_i$ and $\overline{a}_e$.
Such constraints provide an interesting challenge, but we leave direct
understanding of the mathematical structure of their solution to future work.
To make progress, we, therefore, consider restricting to $\gamma$ independent
solutions for the $a_i$ and $\overline{a}_e$.
While not the general class of solutions to \cref{eq:AugmentedSyzygy}, this
strategy is automatic in the $\Delta = 0$ case and we consider it more generally
due to its success there.
This allows us to break down the proportionality constraint into two separate
constraints
\begin{align}
  \label{eq:DeltaProportionalDivergenceConstraint}
	A^{[\Delta, 0]} B^{\Delta} &=  \sum_{i\in{\rm ISPs}(\Gamma)} \partial_i a_i + \sum_{e\in {\rm props}(\Gamma)}z_e \partial_e \bar{a}_e,
  \\
  \label{eq:DeltaProportionalSyzConstraint}
	A^{[\Delta, 1]} B^{\Delta + 1} &= \sum_{i\in {\rm ISPs}(\Gamma)} a_i\partial_i B + \sum_{e \in {\rm props}(\Gamma)} \overline{a}_e z_e \partial_e B,
\end{align}
where we impose that the $A^{[\Delta, i]}$ are polynomial.
Note that, for $\Delta = 0$, the constraint that $A^{[\Delta, 0]}$ is
polynomial, is solved for any values of $a_i$ or $\overline{a}_e$. However, for
higher values of $\Delta$, this constraint is non-trivial.

Let us consider the two proportionality constraints in turn. First, we see that
\cref{eq:DeltaProportionalDivergenceConstraint} depends on the
derivatives of the unknown polynomials. If we take it both on the maximal cut
and the zero locus of the Baikov polynomial, that is, we set $B = z_e = 0$, we see that 
\cref{eq:DeltaProportionalDivergenceConstraint} can be read as a statement that
vector field of the $a_i$ is divergenceless.
The second constraint, \cref{eq:DeltaProportionalSyzConstraint} depends linearly
on the unknown polynomials, and so is again a syzygy constraint.
We can make an analogous analysis to \cref{sec:GeometryToAlgebra} by decomposing
$A^{[\Delta, 1]}$ into an on-shell and off-shell part as
\begin{equation}
	A^{[\Delta, 1]} = a_0^{[\Delta, 1]} + \sum_{e \in \text{props}(\Gamma)} \tilde{a}_0^{[\Delta, 1]} z_e.
\end{equation}
The requirement that the $a_0^{[\Delta, 1]}$ is polynomial therefore implies the
constraint that $a_0^{[\Delta, 1]}$ belongs to the ideal
\begin{equation}
  A_0^{\Gamma, \Delta}  = J_{\text{syz}, \Delta}^{\Gamma}  : \langle B^{1+\Delta} \rangle,
  \label{eq:DeltaIdealMembership}
\end{equation}
where
\begin{equation}
	J_{\text{syz}, \Delta}^{\Gamma} = \left \langle \partial_i B  \,\, : \,\, i \in \text{ISPs}(\Gamma)\rangle + \langle z_e\partial_eB, \, z_e B^{1+\Delta}  \,\, : \,\, e \in \text{props}(\Gamma) \right \rangle.
\end{equation}
Note that this constraint is a generalization of the discussion of
\cref{sec:GeometryToAlgebra}, as $A_0^{\Gamma, 0} = A_0^\Gamma$.
Moreover, note that if we use \cref{eq:AugmentedSyzygy} to build the associated surface
term $\mathcal{N}$, then the large-$\epsilon$, maximal cut limit is $a_0^{[\Delta,
1]}$. Therefore, we again see that, in the large-$\epsilon$, maximal
cut limit, a surface term belongs to an ideal.

We can develop further insight into the meaning of $\Delta$ by
observing that we can rewrite $A_0^{\Gamma, \Delta}$ as\footnote{This follows by the
more general identity that for two $R$-ideals $J$, $K$ and $x$, an element of
$R$, one has that $(J+K x^N) : \langle x^N \rangle = J:x^N + K$, which can be
easily proven by simple two-sided inclusion arguments.}
\begin{equation}
A_0^{\Gamma, \Delta} = \left[ \left \langle \partial_i B  \,\, : \,\, i \in \text{ISPs}(\Gamma)\right\rangle + \left\langle  z_e\partial_eB \,\, : \,\, e \in \text{props}(\Gamma) \right \rangle \right]: \langle B^{1+\Delta} \rangle \,\, + \,\, J_{\text{cut}}^\Gamma.
\label{eq:A0DeltaRewrite}
\end{equation}
This representation of $A_0^{\Gamma, \Delta}$ tells us that $\Delta$ controls the power
of the Baikov polynomial in the quotient of the inner ideal of \cref{eq:A0DeltaRewrite}.
We therefore see that there exists a finite $\overline{\Delta}$ for which
$A_0^{\Gamma, \Delta}$ stabilizes, which can be recognized as the saturation index of the
quotient in \cref{eq:A0DeltaRewrite}.
Moreover, we see that $A_0^{\Gamma, \Delta}$ may be larger than $A_0^\Gamma$ as quotienting by
higher powers may lead to new elements of the ideal.
As we have recognized $A_0^{\Gamma, \Delta}$ as the set of surface terms in the large
$\epsilon$, on-shell limit, we can expect this construction to lead to more
surface terms in cases where $\overline{\Delta} > 0$.
We will see in \cref{sec:tthApplication} that there do exist physical examples
where $\overline{\Delta} \ne 0$, and, correspondingly, that this construction
leads to more surface terms.

\section{One-Loop Critical Syzygies}
\label{sec:oneloop}

Having introduced the theory of critical syzygies in the previous section, here
we begin their practical study by considering them at one loop. At this loop
order, it is well understood that there is at most one master integral
associated to each topology. In this section, we will discuss how this statement
arises by using critical syzygies to explicitly construct surface terms.
An important practical aspect that we study is that of power-counting
constraints. It is well known that the Feynman rules of gauge theory lead to an
upper bound on the total polynomial degree of the numerators that one must
consider in an amplitude calculation.
Specifically, letting $|\Gamma|$ be the number of edges in $\Gamma$,
the numerator associated to $\Gamma$ is a polynomial in
\begin{equation}
  R^{(\Gamma, 1)} = \left\{ \, p \, \in \, R \,\, : \,\, \mathrm{deg}(p) \le |\Gamma| \, \right\},
  \label{eq:OneLoopPowerCountingSpace}
\end{equation}
where $\mathrm{deg}(p)$ is the total polynomial degree of $p$ in Baikov
variables and the index $1$ in $R^{(\Gamma, 1)}$ denotes that it is the
one-loop power-counting space.
Concretely, our aim is to construct a basis of the space of critical surface
terms for Feynman integrals with unit propagator powers that are compatible with
power counting, i.e. a basis of $\text{CSyzSurface}(\Gamma, \vec{1}) \cap
R^{(\Gamma, 1)}$.

To begin our analysis, let us consider the discussion at the end of
\cref{sec:GeometryToAlgebra}, which tells us that non-trivial calculation
is only required if the zero-locus of the $\Gamma$-cut of the Baikov polynomial
is a smooth surface. Explicitly we must check if $U_{\text{sing}}^\Gamma$ is
empty, that is, if there are any solutions to \cref{eq:singularSubLocus} for
$\Gamma_i = \Gamma$.
To understand this, an important observation is to recall
that, at one loop, the Baikov polynomial is at most quadratic in each Baikov
variable.
Without loss of generality, we will order the Baikov parameters such that $z_0,
\ldots, z_{I-1}$ label the $I$ irreducible scalar products and $z_{I}, \ldots,
z_{N-1}$ label the propagators. It is then easy to write the one-loop Baikov
polynomial as
\begin{equation}
  B =
  \frac{1}{2}
  (\vec{z}_i, \vec{z}_e, 1)
  \left(
  \begin{array}{ccc}
      \mathcal{H}_\Gamma & X & \vec{\mathcal{B}_i}  \\
      X^T & O & \vec{\mathcal{B}_e}\\
      \vec{\mathcal{B}_i}^T & \vec{\mathcal{B}_e}^{T} & \mathcal{B}_0
  \end{array}
   \right)
   \left(
   \begin{array}{c}
     \vec{z}_i \\
     \vec{z}_e \\
     1
   \end{array}
   \right),
    \label{eq:BaikovAsQuadric}
\end{equation}
where we gather the ISPs and propagators into $\vec{z}_i$ and $\vec{z}_e$
respectively and the $\mathcal{H}_\Gamma$ and $O$ are symmetric matrices of side-length $I$ and
$N-I$ respectively.
Notice that the matrix $\mathcal{H}_{\Gamma}$ has been defined so that it is the Hessian of the
cut Baikov polynomial, i.e.
\begin{equation}
  \left(\mathcal{H}_{\Gamma}\right)_{ij} = \partial_i \partial_j B|_{\text{cut}_\Gamma}.
\end{equation}
The representation of \cref{eq:BaikovAsQuadric} allows us to easily compute the
partial derivatives of the Baikov polynomial with respect to an ISP as
\begin{equation}
  \vec{\partial_i} B = \mathcal{H}_\Gamma \vec{z}_i + X \vec{z}_e + \vec{\mathcal{B}_i}.
  \label{eq:BaikovPartials}
\end{equation}
It is therefore clear that the critical locus of the $\Gamma$-cut Baikov
polynomial at one loop is given by an intersection of hyperplanes.
The dimensionality of this
configuration is controlled by the rank of the rectangular matrix $(\mathcal{H}_\Gamma \, \, \vec{\mathcal{B}}_i)$.
For the moment, we will proceed with the assumption that this rank is maximal, 
and so the critical locus is given by a single point. In this situation we can
solve for the ISPs as a function of the Baikov derivatives, which allows us to
rewrite the Baikov polynomial on the maximal cut as
\begin{equation}
  B|_{z_e = 0 \, : \, e \in \text{props}(\Gamma)} = \frac{1}{2} (\vec{\partial_i} B, 1)
  \left(
  \begin{array}{cc}
      \mathcal{H}_{\Gamma}^{-1} & 0 \\
      0 & \mathcal{B}_0 - \vec{\mathcal{B}_i}^T \mathcal{H}_\Gamma^{-1} \vec{\mathcal{B}_i}
  \end{array}
   \right)
   \left(
   \begin{array}{c}
     \vec{\partial_i} B \\
     1
   \end{array}
   \right)
    \label{eq:QuadricDerivativeExpression}
\end{equation}
Taking \cref{eq:QuadricDerivativeExpression} on $B = \partial_iB = 0$, we see
that the smoothness condition is translated into the algebraic constraint on the
external kinematics that the constant term in
\cref{eq:QuadricDerivativeExpression} is non-zero.

\subsection{Regular Cases}

Let us begin with a class of critical syzygies that clearly arise without
involved calculation. At one loop, they turn out to generate the full collection
of surface terms in the case where the number of master integrals associated to
$\Gamma$ is 1, and are almost sufficient in the case where the number of master
integrals is 0. They (non-manifestly) contain the set of surface terms in the
OPP basis~\cite{Ossola:2006us} as well as those that allow for reduction of
$\epsilon$-dimensional numerators (see e.g.~\cite{Abreu:2017xsl}).
For each ISP $z_i$, let us consider a syzygy where
\begin{equation}
  a_0 = \partial_i B, \quad a_i = - B, \quad a_{j \ne i} = 0,
  \quad \text{and} \quad
  \tilde{a}_e = \overline{a}_e = 0. \label{eq:trivialSyzygy}
\end{equation}
This clearly solves \cref{eq:MasterSyzygy} as all we have done is to take
anti-symmetric combinations of the generators. Indeed, they are a subset of
``principal syzygy'' solutions to \cref{eq:MasterSyzygy}.
For this reason, we
will denote each such principal critical syzygy as $\vec{a}_i^{p}$.
As the solution set of \cref{eq:MasterSyzygy} has the structure of a
module, we can multiply any solution by a polynomial and still get a solution.
We therefore consider the syzygy $\lambda \vec{a}_i^{p}$ which gives rise to the surface term
\begin{equation}
  S_{\Gamma}(\lambda \vec{a}_i^{p}) = \lambda \alpha_i + \frac{B \partial_i \lambda}{\gamma+1},
\end{equation}
where we implicitly use that all of the $\nu_i = 1$, suppressing the notation
and define
\begin{equation}
  \alpha_i = \partial_i B.
\end{equation}
Due to the fact that, at one loop, the Baikov polynomial is quadratic in all
variables, the $\alpha_i$ are degree 1 in Baikov variables. Moreover, for
generic kinematics they are linearly independent. Therefore, they
form a natural set of variables on the cut associated to $\Gamma$, and we will
phrase our surface term construction in terms of them.

Let us now consider using the principal critical syzygies to build surface
terms. The task is to choose an appropriate set of $\lambda$ such that we have a
basis of the full space of associated surface terms on the cut associated to
$\Gamma$, while remaining in the one-loop power counting space, $R^{(\Gamma,
  1)}$. To this end, let us first observe that, in
the large-$\epsilon$ limit
\begin{equation}
  \lim_{\epsilon \rightarrow \infty}   S_{\Gamma}(\lambda \vec{a}_i^{p}) \in
  J^{\Gamma}_{\text{crit}(B)}.
\end{equation}
We can therefore regard the full $S_{\Gamma}(\lambda \vec{a}_i^{p})$ as a
prescription to lift an element of $J^{\Gamma}_{\text{crit}(B)}$ to a surface
term.
It is easy to see that, given that the $\alpha_i$ are linear in the $z_i$,
$J^{\Gamma}_{\text{crit}(B)}/J^{\Gamma}_{\text{cut}}$ is just the space
of all monomials in the $\alpha_i$ with degree at least 1.
By this argumentation, through an appropriate choice of $\lambda$ and $i$ we can
generate a surface term associated to every such monomial. These surface terms
are clearly linearly independent, as their large-$\epsilon$ limit is a linearly
independent set of monomials in $\alpha_i$. Moreover, as
\begin{equation}
   \mathrm{deg}\left( S_{\Gamma}(\lambda \vec{a}_i^{p})  \right) = \mathrm{deg}\left( \lim_{\epsilon \rightarrow \infty}   S_{\Gamma}(\lambda \vec{a}_i^{p})  \right) = \mathrm{deg}(\lambda) + 1, \label{eq:degSurfacelargeEps}
\end{equation}
we see that, first, the lifting procedure from a monomial to a surface term does
not change the power counting away from that of the large $\epsilon$ limit and,
second, by consideration of \cref{eq:OneLoopPowerCountingSpace}, we have a degree
bound on the monomial $\lambda$ which is easy to satisfy.

To consider this procedure more concretely, let us construct a series of surface
terms for a box diagram $\Gamma$ that is a subtopology of a top-level pentagon.
To this end, we begin by working in the Baikov parameterization of the pentagon, where
$\gamma = -1 - \epsilon$.
Power counting limits us to at most degree
$4$ numerator polynomials for the box. The associated set of surface terms is
then
\begin{equation}
  \left\{  \alpha_0^{n+1} - \frac{n}{\epsilon} B \alpha_0^{n-1} (\partial_0)^2 B \quad : \quad n \in \{0, 1, 2, 3\} \right\},
  \label{eq:trivialSurfaceOneloop}
\end{equation}
where we denote the ISP of the box as $\alpha_0$.
It is clear that all of these functions are linearly independent and that they
are linearly independent of the scalar integral, which we can take as our
master, as expected.

\subsection{Singular Cases}

From the discussion at the top of the section, there are naturally two
situations where the non-singularity of the Baikov polynomial on the maximal cut
comes into question. The first is when the Hessian of the cut Baikov polynomial
is not invertible. The second is when the constant term in
\cref{eq:QuadricDerivativeExpression} vanishes, and hence the Baikov polynomial
gives a singular variety. Let us consider these two cases in turn.

Firstly, we consider the Hessian of the cut Baikov. In order to do this, we
write the Baikov polynomial in a special form, making use of the well-known
``Cayley-Menger'' trick.\footnote{This trick is intimately related to the
  so-called ``Embedding-space formalism''~\cite{Simmons-Duffin:2012juh,Abreu:2017xsl}.}
To employ the Cayley-Menger trick, let us begin by defining the momenta $q_k$
and masses $m_k$ as those associated to the Baikov variable $z_k$. That is,
\begin{equation}
  z_k = (\ell - q_k)^2 - m_k^2.
\end{equation}
This allows us to write the Baikov polynomial as
\begin{equation}
  B =
  \frac{(-1)^{N+1}}{2^N}
  \mathrm{det} \left(
    \begin{array}{ccccc}
      0 & z_0 & \ldots & z_{N-1} & 1 \\
      z_0 & C_{00} & \ldots & C_{(N-1) 0} & 1 \\
      \vdots & \vdots & \ldots & \vdots & \vdots \\
      z_{N-1} & C_{(N-1) 0} & \ldots & C_{(N-1)(N-1)} & 1 \\
      1 & 1 & \ldots & 1 & 0 
    \end{array}
  \right),
  \label{eq:BaikovCayleyMenger}
\end{equation}
where $C_{kl} = (q_k - q_l)^2 - m_k^2 - m_l^2$ is the so-called ``Cayley
matrix''. The benefit of this notation is that it allows us to write the entries
of the matrices in \cref{eq:BaikovAsQuadric} as minors of the matrix
\begin{equation}
  \mathcal{C} =
  \left(
  \begin{array}{cc}
    C & \vec{1} \\
    \vec{1}^T & 0
  \end{array}
  \right),
\end{equation}
where the $\vec{1}$ represents a column containing just $1$s. Note that, again
by the Cayley-Menger trick, we have that
\begin{equation}
  G(q_1, \ldots q_{N-1}) = \frac{(-1)^{N}}{2^{N-1}} \det(\mathcal{C}),
\end{equation}
i.e. it is the Gram determinant associated to the Baikov parameterization. We
will denote the minor of $\mathcal{C}$ where row $i$ and column $j$ have been
removed as $\mathcal{C}[\hat{i}, \hat{j}]$.
With this notation, differentiation of \cref{eq:BaikovCayleyMenger} gives
\begin{equation}
  \left( \mathcal{H}_{\Gamma} \right)_{ij} = \frac{(-1)^{i+j+N}}{2^{N-1}} \mathcal{C}[\hat{i}, \hat{j}], \qquad i, j \in [0, \ldots, I-1].
\end{equation}
To understand if $\mathcal{H}_\Gamma$ is invertible, we compute its determinant. By Jacobi's
theorem on complementary minors this is given by
\begin{equation}
  \det(\mathcal{H}_\Gamma) = \left[\frac{(-1)^N}{2^{N-1}}\right]^{I}\det(\mathcal{C})^{I-1} G_\Gamma,
\end{equation}
where
\begin{equation}
    G_\Gamma = \det\left(\mathcal{C}_{ef}, \qquad e, f \in [I, \ldots, N-1, N]\right)
\end{equation}
is the determinant of the minor of $\mathcal{C}$ corresponding to
the propagators of $\Gamma$. By the Cayley-Menger trick, we have that this
determinant is a constant multiple of the Gram determinant that we associate to
$\Gamma$.\footnote{We define $G_\Gamma$ of a tadpole graph, that involves no momenta, to be 1.}
The case where $\det(\mathcal{H}_\Gamma)$ vanishes because $\det(\mathcal{C})$
vanishes is not of interest as, looking to \cref{eq:BaikovRepresentation}, we
see that it corresponds to a region of phase-space where the Baikov
parameterization itself is degenerate.
We therefore conclude that the critical locus of the Baikov polynomial at one
loop can only be non-isolated if the Gram determinant of the external momenta of
$\Gamma$ is zero. For fixed-angle scattering at one loop, this can occur in only
one case where the associated integrals are not scaleless: the ``external leg
correction''-bubble that we depict in \cref{fig:OneLoopNonIsolated}.
In practice, it turns out that all maximal minors of the rectangular matrix
$(\mathcal{H}_\Gamma \,\, \vec{B}_i)$ vanish and hence the critical locus of
this topology is not a finite set of points.
In such cases, critical syzygies are insufficient for a complete reduction to
master integrals. We therefore leave further study to future work.
\begin{figure}
\includegraphics[scale=1.5]{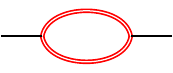}
  \centering
  \caption{The single one-loop graph in fixed-angle scattering for which the critical locus
           is not isolated. The red lines represent a field of mass $m$ (e.g. a
           top). The black lines represent massless particles, that is $p^2 =
           0$. The associated Feynman integral is only not scaleless for $m \ne 0$.}
  \label{fig:OneLoopNonIsolated}
\end{figure}

If the Hessian of the Baikov matrix is of full rank on the maximal cut, then the
critical locus is isolated. As stated earlier, it remains to check the
non-singularity condition. To this end, let us rewrite the constant term
of~\cref{eq:QuadricDerivativeExpression} by recognizing it as a Schur complement
as
\begin{equation}
\det\left(
  \begin{array}{cc}
    \mathcal{H}_\Gamma & \vec{\mathcal{B}_i} \\
    \vec{\mathcal{B}}_i^T & \mathcal{B}_0
  \end{array}
\right)
= \left( \mathcal{B}_0 - \vec{\mathcal{B}_i}^T \mathcal{H}_\Gamma^{-1} \vec{\mathcal{B}_i} \right) \det(\mathcal{H}_\Gamma).
\end{equation}
It therefore remains to compute the determinant on the left-hand-side. This is
easily achieved using the same logic as the computation of $\det(\mathcal{H}_\Gamma)$, only
noting that we now include the last row and column of $\mathcal{C}$, hence we
have that
\begin{equation}
\det\left(
  \begin{array}{cc}
    \mathcal{H}_\Gamma & \vec{\mathcal{B}_i} \\
    \vec{\mathcal{B}}_i^T & \mathcal{B}_0
  \end{array}
\right)
= \left[\frac{(-1)^N}{2^{N-1}}\right]^{I+1} \det(\mathcal{C})^I C_\Gamma ,
\end{equation}
where
\begin{equation}
C_\Gamma = \det( C_{ef}, \qquad e,f \in [I, \ldots, N-1]),
\end{equation}
is the minor of the Cayley-matrix $C$, associated to the propagators that are
cut. Altogether, we find that we can write
\begin{equation}
\mathcal{B}_0 - \vec{\mathcal{B}_i}^T \mathcal{H}_\Gamma^{-1} \vec{\mathcal{B}_i} = \frac{(-1)^{N}}{2^{N-1}} G(q_1, \ldots, q_{N-1}) \frac{C_\Gamma}{G_\Gamma}.
\end{equation}
We therefore see that the zero locus of the $\Gamma$-cut Baikov polynomial is a
singular variety either if $G(q_1, \ldots, q_{N-1}) = 0$ or if $C_\Gamma = 0$.
As earlier, we discard the zero-Gram case and consider the more interesting case,
$C_\Gamma = 0$. It is well-known that this corresponds to
the occurrence of first-type Landau singularities (see, e.g.~\cite{Mizera:2021icv} for recent
discussion). Examples of diagrams that exhibit this phenomenon are infra-red
divergent triangle graphs, such as those that we exemplify in
\cref{fig:VanishingCayleyExamples}.
\begin{figure}
    \centering
\includegraphics[scale=1.5]{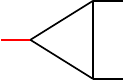}\hspace{3em}
\includegraphics[scale=1.5]{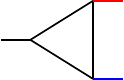}\hspace{3em}
\includegraphics[scale=1.5]{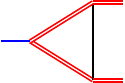}
  \caption{Representative examples of one-loop diagrams which have vanishing
    Cayley determinant. Black lines represent massless particles. Doubled red
    lines represent massive particles. Blue or single-red lines represent
    off-shell particles.
    For these examples, the set of critical syzygies is
    larger than the set generated by principal syzygies and further computation
    is required.}
  \label{fig:VanishingCayleyExamples}
\end{figure}
Here, the set of critical syzygies becomes larger than the set generated by
principal syzygies.
Specifically, as the critical locus is already a point, we expect that there is
exactly one interesting non-principal syzygy. Concretely, the ideal generated by
the partial derivatives does not contain the constant polynomial, whereas
$A_0^\Gamma$ does, hence there must be a syzygy where $a_0$ is a
constant. This is most easily displayed in the Baikov representation associated
to the cut, that is, the one where there are no ISPs. Considering the case of
the one-mass triangle in \cref{fig:VanishingCayleyExamples}, we explicitly find
the syzygy
\begin{equation}
0 = a_0 B + \overline{a}_1 z_1 \partial_1 B + \overline{a}_2 z_2 \partial_2 B + \overline{a}_3 z_3 \partial_3 B + \tilde{a}_2 z_2 B,
\end{equation}
where the components of the syzygy are given by
\begin{align}
  \begin{split}
    a_0 &= -2 s^2, \\
    \overline{a}_1 &=  s^2 - s [z_1 - z_3] + 4 z_2[2 s - 2 z_3 + z_2 - z_1], \\
    \overline{a}_2 &= 2 s^2 + 6 s z_2 + 4 z_2^2 + s (z_1 + z_3) - 4 [z_1 z_2 + z_3 z_2 + z_1 z_3], \\
    \overline{a}_3 &= s^2 - s [z_3 - z_1] + 4 z_2[2 s - 2 z_1 + z_2 - z_3 ], \\
    \tilde{a}_2 &= - 8 (s + z_2 - z_1 - z_3).
    \end{split}
\end{align}
It would be interesting to understand if the full syzygy could be determined
from geometrical arguments, but we leave this to future work.
This syzygy then gives rise to a surface term,
\begin{equation}
  S_\Gamma(\vec{a})
  = - 2 s^2 - 8 z_2 (s + z_2 - [z_1  + z_3])\left( 1-\frac{1}{\gamma} \right) + \frac{s}{\gamma} (z_1 + z_3 + 2 z_2),
\end{equation}
which manifestly becomes constant on the triangle cut $z_1 = z_2 = z_3 = 0$, and
hence provides a reduction relation for the scalar triangle.

\section{Critical Surface Terms at Two Loops}
\label{sec:twoloop}

The construction of critical surface terms at two loops is a more complicated
affair than at one loop. At two loops, the structure of the Baikov polynomial is
more involved than a quadric, and so we are more constrained in our
ability to analytically construct critical surface terms.
In this section, we study the question: can we computationally use critical
surface terms to perform reduction to master integrals?
While surface terms from standard syzygies have
been found to be sufficient in many calculations, as critical surface terms are
effectively a subset, this question should be studied.
As at one loop, in practical applications to two-loop integral
reduction, surface terms are constructed satisfying power-counting constraints.
For concreteness, we again adopt gluonic power counting, defining the two-loop
power-counting space of numerator polynomials associated to a diagram $\Gamma$
as
\begin{align}
  R^{(\Gamma, 2)} = \big\{a \,\in \,R \quad : \quad {\rm deg}_{\ell_1}(a) \leq |\Gamma_1|, \,\,\, {\rm deg}_{\ell_2}(a) \leq |\Gamma_2|, \,\,\, {\rm deg}_{\ell_1}(a) + {\rm deg}_{\ell_2}(a) < |\Gamma| \big\}\,,
  \label{eq:TwoLoopPowerCountingSpace}
\end{align}
where $|\Gamma_{i}|$ is the the number of edges in $\Gamma$ that depends on loop
momentum $i$, $|\Gamma|$ is the total number of edges in $\Gamma$ and
$\mathrm{deg}_{\ell_i}$(a) is the polynomial degree of $a$ in loop momentum $i$.
In order to reduce the integrals arising in a QCD amplitude that are associated
to the set of propagators $\Gamma$ and exponents $\vec{\nu}$, an integral
reduction program needs to explicitly construct elements of
\begin{equation}
  \text{Surface}^R(\Gamma, \vec{\nu}) = \text{Surface}(\Gamma, \vec{\nu}) \cap R^{(\Gamma, 2)},
  \label{eq:RenormalizableSurface}
\end{equation}
surface terms that live within power counting.
The goal of this section is to study computational construction of surface terms
from critical syzygies.
Naturally, the isomorphism in \cref{eq:CriticalSurfaceSufficiency} implies that, if
a diagram has isolated critical points and $\mu=1$, critical surface terms are
in principle sufficient to reduce all tensor integrals associated to this
diagram.
However, it is {\em a priori} unclear if it is possible to construct a
finite-dimensional restriction of $\text{CSyzSurface}(\Gamma, \vec{\nu})$ that
lives within power counting in a
finite number of steps.
That is, can we computationally construct a set of critical syzygies that give
rise to a subspace $\text{CSyzSurface}^R(\Gamma, \vec{\nu}) \subset
\text{Surface}^R(\Gamma, \vec{\nu})$ such that
\begin{align}
  \text{CSyzSurface}^R(\Gamma, \vec{\nu})
  \simeq
  \text{Surface}^R(\Gamma, \vec{\nu}) / [\text{Surface}^R(\Gamma, \vec{\nu}) \cap J^\Gamma_{\text{cut}}]?
  \label{eq:isomorphismWithPC}
\end{align}
In this section, we study this question in non-trivial examples and thereby
provide evidence that critical surface terms are a useful tool for the reduction
of tensor integrals in gauge theory.
Unlike in the one-loop case, generic analytic construction of surface terms is a
non-trivial problem. Indeed, at two loops, one experimentally finds that $B=0$
on the maximal cut is a singular variety. That is, $U_{\text{sing}}^\Gamma$ is
never empty, and, therefore, principal syzygies are insufficient.
For this reason, to study the question posed in \cref{eq:isomorphismWithPC}, we
set up an algorithm to determine critical syzygies within power counting.
We then consider a cutting-edge example: the leading-color contributions to the
two-loop $t\bar{t}H$ amplitude.
We demonstrate that, for the light-quark contributions (whose master integrals
were recently computed~\cite{FebresCordero:2023pww}) the critical locus is
always a finite set of points and that the $\mu = 1$ hypothesis holds in all but
one case. In this way we are able to provide a positive answer to the question
raised in \cref{eq:isomorphismWithPC} in a physically interesting example.

\subsection{Computational Construction of Critical Surface Terms}\label{sec:twoloopTechnique}

Let us now turn to the question of constructing a basis of critical surface
terms that are compatible with the power-counting constraints.
The practical strategy that we will employ is to take the definition of
$\text{CSyzSurface}(\Gamma, \vec{\nu})$ in \cref{eq:CSyzSurfaceDefinition} as a
prescription to construct critical surface terms.
Specifically, we will explicitly construct a finite subspace
$\text{CSyz}^R(\Gamma) \subset \text{CSyz}(\Gamma)$, and define
\begin{equation}
	\text{CSyzSurface}^R(\Gamma, \vec{\nu}) = \{ S_\Gamma(\vec{a}, \vec{\nu}) : \vec{a} \in \pi(\text{CSyz}^R(\Gamma) ) \}.
\end{equation}
Furthermore, we will constrain the basis of $\pi(\text{CSyz}^R(\Gamma))$ such
that all surface terms are within power counting.
An important observation is that this construction of
$\text{CSyzSurface}^R(\Gamma, \vec{\nu})$ develops a strong dependence on the
choice of $\pi$.
To see this, consider specifying $\pi$ by specifying a basis $\{[\vec{a}_1], \,
\ldots \}$ of $\text{CSyz}(\Gamma)$ and their associated lifts $\vec{a}_i \in
\text{Syz}(\Gamma)$.
Note that, for any choice of lift $\vec{a}$ of some $[\vec{a}]$, the large-$\epsilon$, on-shell limit
of $S_\Gamma(\vec{a}, \vec{\nu})$ is independent of the specific choice.
However, the polynomial degree of $S_\Gamma(\vec{a}, \vec{\nu})$ does depend on
the specific choice of the lift, through the sub-leading terms in the large-$\epsilon$, on-shell limit.
We can therefore see that to computationally construct $\text{Surface}^R(\Gamma,
\vec{\nu})$ via critical syzygies will require carefully choosing the lifts.

A conceptually simple way to construct $\text{CSyz}^R(\Gamma)$ is to write an
ansatz for some $\vec{a}$ that is a syzygy whose associated surface term lives
within power counting. One then uses linear algebra to find a basis of such
terms which are linearly independent in the large-$\epsilon$, on-shell limit.
Nevertheless, due to the high polynomial degrees involved, this is a
computationally demanding approach.
Instead, we construct an approach to generate critical surface terms which
exploits the module properties of $\text{CSyz}(\Gamma)$, thereby keeping the
computation tractable.
Specifically, we break the problem down into two steps. We first construct a
generating set of $\text{CSyz}(\Gamma)$ using linear algebra methods given by
low-degree representatives in $\text{Syz}(\Gamma)$.
We then construct a basis of $\text{CSyz}^R(\Gamma)$ by taking polynomial
combinations of our generating set, constraining the combinations such the
associated surface terms satisfy the power-counting constraints and are linearly
independent in the large-$\epsilon$, on-shell limit.
In the following, we elaborate on the details of each component of our approach,
highlighting tricks for reducing the size of the involved linear systems. We
provide a summary of the approach in \cref{sec:algorithm}.

\subsubsection{Constructing a Generating Set of Critical Syzygies}
\label{sec:CriticalSyzGenerators}

Let us consider the question of constructing a generating set of the module of
critical syzygies, $\text{CSyz}(\Gamma)$, defined in \cref{eq:CSyzDefinition}.
Our approach will be to use linear algebra to construct a generating set
$\{[\vec{v}_1], \ldots \}$ of
$\text{CSyz}(\Gamma)$ by finding representatives $\vec{v}_i$ of these
elements in $\text{Syz}(\Gamma)$. Note that a generating set of a module
is not unique and it is a non-trivial problem to construct a minimal generating
set of a module.
Instead, we will construct a generating set of
$\text{CSyz}(\Gamma)$ up to a input power-counting bound. This is naturally
not a minimal construction and, therefore, the number of generators will depend on this bound.
We will then check if this set generates $\text{CSyz}(\Gamma)$ by exploiting the
isomorphism of $\text{CSyz}(\Gamma)$ to $A_0^\Gamma / J^\Gamma_{\text{cut}}$.

We begin by expressing all unknowns of the syzygy equation, \cref{eq:MasterSyzygy}, as polynomials in
Baikov variables with coefficients that are rational functions of kinematics.
Writing all unknowns in \cref{eq:MasterSyzygy} as some $a_j$ for simplicity,
we parameterize them as
\begin{equation}
\label{eq:criticalsyzygyAnsatzGeneral}
   a_j = \sum_{k:|\vec{\alpha}_k|\leq N_j} c_{jk}(\vec{s}) Z^{\vec{\alpha}_k},
\end{equation}
where we write a monomial in Baikov variables as
\begin{equation}
  Z^{\vec{\alpha}}:=\prod_{m=1}^{N} z_m^{\alpha_{m}}, \quad \alpha_{m}\in \mathbb{Z}_{\ge 0},
\end{equation}
$|\vec{\alpha}| = \sum_m \alpha_{m}$ is the total degree and the coefficients
$c_{jk}$ are unknown rational functions in external kinematics
$\vec{s}$.\footnote{While $\text{CSyz}(\Gamma)$ is an $R$-module, and so elements
of $\text{CSyz}(\Gamma)$ can depend on $\epsilon$, the syzygy equation
\cref{eq:MasterSyzygy} is independent of $\epsilon$ and so the generating set
need not depend on $\epsilon$. Hence, choosing the $c_{jk}$ to be $\epsilon$ independent is no restriction.}
In \cref{eq:criticalsyzygyAnsatzGeneral}, we denote the total degree of $a_j$ in Baikov
variables as $N_j \in \mathbb{Z}_{\ge 0}$.
The set of $N_j$ are the input power-counting bounds to our procedure and
control its computational cost.
A useful computational observation is that there are a number of easily
identifiable syzygies $\vec{a}$ which lead to $S_\Gamma(\vec{a}, \vec{\nu}) =
0$. Specifically, this is the set of syzygies, $\text{ZSyz}(\Gamma)$ discussed
in \cref{sec:CriticalSurfaceTerms}.
Such syzygies can be seen as an artefact of our decision to break up the
coefficient of the Baikov polynomial in \cref{eq:MasterSyzygy} into an on-shell part, $a_0$ and
off-shell parts, the $\tilde{a}_e$, which was useful for the analysis
of~\cref{sec:GeometryToAlgebra}. Nevertheless, computationally, such syzygies
are
irrelevant and we remove them from the ansatz by letting $a_0$ depend only on
ISPs and $\tilde{a}_e$ depend on ISPs and only on propagators $z_{e'}$ such that $e' \le e$.

By inserting the ansatz \cref{eq:criticalsyzygyAnsatzGeneral} into the syzygy
equation \cref{eq:MasterSyzygy} and requiring that the coefficient of each
monomial in Baikov variables should vanish, we convert a linear system with
polynomial unknowns $a_k$ into a linear system of unknown rational functions
$c_{jk}$.
That is, gathering all unknowns $c_{jk}$ of
\cref{eq:criticalsyzygyAnsatzGeneral} into the object $C= \bigcup_{j,k} c_{jk}$,
we reparametrize the syzygy equation \cref{eq:MasterSyzygy} as
\begin{align}
  \sum_j M_{ij} C_j = 0,
  \label{eq:criticalSyzygyEqsys}
\end{align}
where $M$ is a matrix of rational functions of the kinematics and the index $i$
runs over all monomials in Baikov variables that arise when inserting the ansatz
into \cref{eq:MasterSyzygy}.
By \cref{eq:criticalsyzygyAnsatzGeneral}, a basis of solutions to
\cref{eq:criticalSyzygyEqsys} corresponds to a basis of the degree-bounded
subspace of $\text{Syz}(\Gamma)$ controlled by the $N_j$, which we denote as
$\{\vec{v}_1, \ldots, \vec{v}_{C(\vec{N})}\}$.
Naturally, each $\vec{v}_i$ is a representative of some $[\vec{v}_i]$ of
$\text{CSyz}(\Gamma)$. We note that, as representatives $\vec{v}$ that differ by
an element with zero critical part are equivalent, the elements $\{[\vec{v}_1],
\ldots, [\vec{v}_{C(\vec{N})}]\}$ are not necessarily linearly independent.
Nevertheless, we choose to leave these redundancies in our system, as they will
be tackled at the next stage.

The set $\{[\vec{v}_1], \ldots, [\vec{v}_{C(\vec{N})}]\}$ forms a natural guess for a generating
set of the unbounded $\text{CSyz}(\Gamma)$ as, for sufficiently high $\vec{N}$ it
must generate the module, by the ascending chain condition. In order to check if
we indeed do have a generating set of $\text{CSyz}(\Gamma)$ we make use of the
fact that it is isomorphic to $A_0^\Gamma/J_{\text{cut}}^\Gamma$.
Specifically, the set $\{[\vec{v}_1], \ldots, [\vec{v}_{C(\vec{N})}]\}$
generates $\text{CSyz}(\Gamma)$ if and only if the critical parts of the
$\vec{v}_i$ generate $A_0^\Gamma/J_{\text{cut}}^\Gamma$. Equivalently,
\begin{equation}
     \langle [\vec{v}_1], \ldots, [\vec{v}_{C(\vec{N})}] \rangle = \text{CSyz}(\Gamma) \qquad \Longleftrightarrow \qquad \langle \mathfrak{c}(\vec{v}_1), \ldots, \mathfrak{c}(\vec{v}_{C(\vec{N})}) \rangle + J_{\text{cut}}^\Gamma = A_0^\Gamma.
    \label{eq:CriticalGenerationTest}
\end{equation}
To check if we indeed have a generating set of $\text{CSyz}(\Gamma)$, we can
therefore check equality of the two ideals on the right-hand side of
\cref{eq:CriticalGenerationTest}.
This can easily be performed by checking if their reduced Groebner bases are
equal.
For the ideal generated by critical parts of syzygies, this is easily performed
in the computer algebra system \texttt{Singular}~\cite{DGPS} as we have the
generating set by construction. In order to compute a Groebner basis of
$A_0^{\Gamma}$, we perform the ideal quotient of \cref{eq:A0AsQuotient}
computationally, again using \texttt{Singular}.
If the two Groebner bases are distinct then we do not have a generating set of
$\text{CSyz}(\Gamma)$, reflecting that the degree bounds $\vec{N}$ are too low.
In this way, given $\vec{N}$ we construct a generating set of
$\text{CSyz}(\Gamma)$ or report that the degree bound is too low.

Having set up an algorithm for computing a generating set of
$\text{CSyz}(\Gamma)$, let us make some practical remarks.
While it is trivial to obtain $M$ analytically, solving the system of equations
analytically is highly non-trivial due to the large size of $M$.
For this reason, we apply our approach numerically, at a given phase-space
point.
This approach allows for a number of optimizations that are frequently applied
in finite-field-based approaches, which we record here for completeness.
An important feature is that many of the $c_{jk}$ in
\cref{eq:criticalsyzygyAnsatzGeneral} are zero. This can be detected when
solving \cref{eq:criticalsyzygyAnsatzGeneral} on an initial, randomly-chosen, phase-space point $\vec{s}_0$
and imposing this constraint for later evaluations. That is, we impose
\begin{equation}
    c_{jk}(\vec{s_0}) = 0 \qquad \Rightarrow \qquad c_{jk}(\vec{s}) = 0,
\end{equation}
and, therefore, in practice we use the refined ansatz
\begin{align}
  a_{j} = \sum_{k : c_{jk}(\vec{s}_0) \neq 0} c_{jk}(\vec{s}) Z^{\vec{\alpha}_k}.
  \label{eq:skeletonA}
\end{align}
In \cref{eq:criticalSyzygyEqsys} this has the effect of removing the columns of
$M$ that correspond to $C_j(\vec{s}_0) = 0$.
Moreover, we decrease the number of unknowns even further by observing that
the $c_{jk}$ (or equivalently the $C_j$) are $\mathbb{Q}$-linearly dependent.
Analogous to the approach applied when reconstructing scattering
amplitudes~\cite{Abreu:2019odu}, with a small number of evaluations of the
$C_j(\vec{s})$, we resolve the $\mathbb{Q}$-linear dependencies and write
\begin{align}
C_j = \sum_k A_{jk} \tilde{C}_k,
\end{align}
where the entries of matrix $A$ are rational numbers and $\tilde{C}_k$ are a
linearly independent subset of the $C_j$.
Combining with \cref{eq:criticalSyzygyEqsys} leads to
\begin{align}
  \sum_{l} \tilde{M}_{jl} \tilde{C}_l = 0, \qquad \text{where} \qquad \tilde{M}_{jl} = \sum_{k: C_k(\vec{s}_0) \neq 0} M_{jk}A_{kl} .
  \label{eq:criticalSyzygyEqsysRemoveLD}
\end{align}
That is, one only has to row reduce the matrix $\tilde{M}_{jl}$. These
optimizations can result in linear systems that are hundreds of times smaller
than \cref{eq:criticalSyzygyEqsys}.

In summary, given a set of power-counting bounds $N_j$, we
construct a generating set for $\text{CSyz}(\Gamma)$ on a fixed phase-space
point $\vec{s}$.
The size of the generating set depends on $\vec{N}$ and is neither a Groebner
basis, nor a minimal generating set.
The generating set is specified analytically, though implicitly, through the
linear equation system \cref{eq:criticalSyzygyEqsysRemoveLD}, which can easily
be solved numerically.

\subsubsection{Critical Surface Terms Within the Power-Counting Window}
\label{sec:powercounting}

Using the algorithm of the previous subsection, for each diagram and an
appropriately chosen power-counting bound $\vec{N}$, we can determine a
generating set of $\text{CSyz}(\Gamma)$.
It therefore remains to use these to construct a basis of
$\text{CSyzSurface}^R(\Gamma)$.
The approach we will take to this problem is to build surface terms from
polynomial combinations of our generators set of $\text{CSyz}(\Gamma)$.
That is, we look for syzygies $\vec{w}$ such that
\begin{equation}
    \vec{w} = \sum_{j=1}^{C(\vec{N})} \lambda_j \vec{v}_j \qquad \text{and} \qquad S_\Gamma(\vec{w}, \vec{\nu})  \quad \in \quad R^{(\Gamma, 2)},
    \label{eq:PowerCountingCompatibility}
\end{equation}
where the $\lambda_j$ are polynomials in $R$ and the $\vec{v}_j$ are
representatives in $\text{Syz}(\Gamma)$ of the $[\vec{v}_j]$ in
$\text{CSyz}(\Gamma)$. As the $\vec{v}_j$ are known,
\cref{eq:PowerCountingCompatibility} is a non-trivial constraint on the
polynomials $\lambda_j$.
In this section we discuss two methods for satisfying this constraint.
In the first, we take monomial multiples of each generator $\vec{v}_j$
that satisfy \cref{eq:PowerCountingCompatibility}, finding that this is often
sufficient to find a basis of surface terms.
In the second, we write an ansatz for the $\lambda_j$ and then use linear
algebra to find $\vec{w}$ such that all terms which violate power counting vanish,
analogous to the approach used in ref.~\cite{Abreu:2023bdp}.

To begin, we parameterize our linear combination of the generators $\vec{w}$ as
\begin{equation}
	\vec{w} = \sum_{j,k} \tilde{w}_{jk}(\vec{s}) Z^{\vec{\alpha_j}} \vec{v}_k.
  \label{eq:SyzygyCombinationAnsatz}
\end{equation}
Here, the sum over $j$ runs over a suitably large set of monomials in the Baikov
variables. For practical purposes, we take this to be the set of monomials in
ISPs that satisfy the power-counting constraints.
Given the ansatz in \cref{eq:SyzygyCombinationAnsatz}, we then consider
constructing the surface term $S_{\Gamma}(\vec{w}, \vec{\nu})$. This 
can be written as a linear combination of monomials that live within power
counting and those that do not. That is,
\begin{align}
  S_\Gamma\left( \vec{w}, \vec{\nu} \right) =
  \sum_{Z^{\vec{\alpha}_i} \in R^{(\Gamma, 2)}} n_i(\vec{s}, \gamma, \tilde{w})Z^{\vec{\alpha}_i}
  + \sum_{Z^{\vec{\alpha}_i} \not\in R^{(\Gamma, 2)}} \sum_{j, k} m_{ijk}(\vec{s}, \gamma) \tilde{w}_{jk} Z^{\vec{\alpha}_i},
  \label{eq:twoloopSurfaceWithPC}
\end{align}
where the $m_{ijk}$ and $n_i$ are rational functions of kinematics and $D$ and
the $n_i$ are linear
in the $\tilde{w}_{jk}$.
In order to live within power counting, each term in
\cref{eq:twoloopSurfaceWithPC} that is not in $R^{(\Gamma, 2)}$ must vanish.
That is, we have
\begin{equation}
  \sum_{j, k} m_{ijk}(\vec{s}, \gamma) \tilde{w}_{jk} = 0,
  \label{eq:SyzCombinationPCConstraint}
\end{equation}
which is a linear constraint on the ansatz parameters $\tilde{w}_{jk}$.

We consider finding a basis of critical syzygy solutions to
\cref{eq:SyzCombinationPCConstraint} in two separate ways. The first approach is
combinatorical. Specifically, we enumerate all
possible values of $j$ and $k$, setting $\tilde{w}_{jk} = 1$ and otherwise letting the
entries of $\tilde{w}$ be zero.
We then check if this $\tilde{w}$ satisfies \cref{eq:SyzCombinationPCConstraint}.
Effectively, this amounts to taking $\vec{w} = Z^{\vec{\alpha_j}} \vec{v}_k$ and
checking if it satisfies the conditions
\begin{equation}
  \begin{split}
  w_0 + \sum_{e\in {\rm props}(\Gamma)}\tilde{w}_ez_e \quad \in \quad R^{(\Gamma, 2)},
    \\
  \sum_{i\in {\rm ISPs}(\Gamma)}\partial_i w_i + \sum_{e \in \text{props}(\Gamma)} z_e \partial_e \bar{w}_e \quad \in \quad R^{(\Gamma, 2)}.
  \label{eq:relationNjandGamma}
  \end{split}
\end{equation}
In this way, we generate a collection of $\vec{w}$ which satisfy the
power-counting bounds. Naturally, the associated set of $[\vec{w}]$ exhibit
linear dependencies. We therefore select from our collection a subset with
linearly independent $\mathfrak{c}(\vec{w})$, which is a simple linear
algebra problem.
Importantly, as $S_\Gamma(\vec{w}, \vec{\nu}) \rightarrow \mathfrak{c}(\vec{w})$
in the large-$\epsilon$, on-shell limit, we see that this allows us to count how many
independent surface terms associated to $\Gamma$ we have constructed.
If this is equal to the number of surface terms on the maximal cut, we conclude
that we have found a basis set of solutions to
\cref{eq:SyzCombinationPCConstraint} and therefore found a basis of
$\text{Surface}^R(\Gamma, \vec{\nu})$ modulo pinches and so our set of critical syzygies
is complete.

Naturally, this combinatorical approach of solving
\cref{eq:SyzCombinationPCConstraint} is not guaranteed to find all solutions, as
it may be necessary to take non-trivial linear combinations of products of
monomials and $\text{CSyz}(\Gamma)$ generators. Therefore, if the completeness
test fails, we consider directly solving \cref{eq:SyzCombinationPCConstraint} as
a linear system.
In order to easily identify a linearly independent set of critical surface terms
from our solutions, we additionally require that the $\tilde{w}_{jk}$ are
independent of
$\epsilon$.\footnote{Importantly, it can be shown that a basis of the $\epsilon$-independent
solutions of \cref{eq:SyzCombinationPCConstraint} is also a basis of the
$\epsilon$-dependent solutions of \cref{eq:SyzCombinationPCConstraint}, so this
practical trick causes no conceptual issues.}
This again allows us to certify sufficiency of the basis of surface terms by
looking only at the critical part of the involved syzygy.
Writing $m_{ijk} = m_{ijk}^{(0)} + \frac{1}{\gamma} m_{ijk}^{(1)}$, we therefore require
that the $\tilde{w}_{jk}$ satisfy
\begin{equation}
    \sum_{j,k} m^{(0)}_{ijk}(\vec{s}) \tilde{w}_{jk} = 0, \qquad \sum_{j,k} m^{(1)}_{ijk}(\vec{s}) \tilde{w}_{jk} = 0,
    \label{eq:PowerCountingConstraints}
\end{equation}
where we recall the $m_{ijk}^{(X)}$ are matrices rational in the external
kinematics.
We solve \cref{eq:PowerCountingConstraints} via linear algebra methods which
leads to a collection of syzygies $\vec{w}$ that live within power counting. We
then select from this set of syzygies, a subset with linearly independent
critical part, thereby finding a basis of $\text{CSyz}^R(\Gamma)$ and consequently,
$\text{CSyzSurface}^R(\Gamma, \vec{\nu})$.

Analogous to the construction of the generating set of $\text{CSyz}(\Gamma)$ discussed
in \cref{sec:CriticalSyzGenerators}, the linear equation system in
\cref{eq:PowerCountingConstraints} can pose a non-trivial challenge to solve. We
therefore close this subsection by discussing a number of computational optimizations.
First, note that, in practice, the $m_{ijk}^{(X)}$ are constructed from a set of
generators of $\text{CSyz}(\Gamma)$ that are known on a numerical phase-space
point $\vec{s}^{(0)}$. We therefore discuss solving \cref{eq:PowerCountingConstraints}
at $\vec{s} = \vec{s}^{(0)}$.
Next, we note that the matrices in \cref{eq:PowerCountingConstraints} are sparse
and so we apply sparse linear
algebra methods to solve them.
The sparsity of the solution basis can therefore depend on the order in which
the equations in \cref{eq:PowerCountingConstraints} are solved.
Practically, we first solve one of the two $m^{(X)}$ systems and insert the
solutions into the other in order to solve the full system, choosing the
ordering based upon the sparsity of the final solution.

Having evaluated a basis of solutions to \cref{eq:PowerCountingConstraints} on a
single phase-space point, we are able to use the structure of the resulting
basis (and the sparsity properties that it inherited from our construction
method) to ease its evaluation on further phase-space points.
Letting the $\{ \tilde{w}^{(1)}, \ldots, \tilde{w}^{(M)} \}$ be a basis of solutions to
\cref{eq:PowerCountingConstraints}, we make two structural observations about
the basis.
First, it is clear that rotating the basis by any element of $\text{GL}(M)$ will lead
to a second basis $\{\tilde{w}'^{(1)}, \ldots, \tilde{w}'^{(M)}\}$ that also
satisfies \cref{eq:PowerCountingConstraints}.
One is only able to uniquely fix a basis after specifying this $\text{GL}(M)$
freedom, which can be achieved by choosing an appropriate sub-matrix to be the
identity matrix.
In practice, an appropriate such choice is automatically performed when solving
for the $w_{jk}^{(l)}(\vec{s}_0)$ numerically.
Therefore, we interpret the structure of zeros and ones in
$\tilde{w}^{(l)}(\vec{s}_0)$ as fixing this $\text{GL}(M)$ freedom in the phase-space
independent basis.
Secondly, many of the entries of the evaluation $\tilde{w}_{jk}^{(l)}(\vec{s}_0)$ are
either zero, or identical. It is natural to interpret these zeros/equalities
as phase-space independent.
Altogether, we are able to write an ansatz for the $\text{GL}(M)$ fixed basis as
\begin{equation}
   \tilde{w}_{jk}^{(l)} = \tilde{w}_{jk}^{(l)[0]} + \sum_{m=1}^e y_m(\vec{s}) \tilde{w}_{jk}^{(l)[m]},
   \label{eq:BasisAnsatz}
\end{equation}
where $e$ is the number of distinct, non-zero (and not equal to 1) entries
of $\tilde{w}_{jk}^{(l)}(\vec{s}_0)$ and the entries of each
$\tilde{w}_{jk}^{(l)[m]}$ are either zero or one.
Inserting the basis ansatz \cref{eq:BasisAnsatz} into the power-counting
constraints \cref{eq:PowerCountingConstraints}, we end up with a linear system
for the $y(\vec{s})$ as
\begin{equation}
  \sum_{m = 1}^e A_{nm} y_m(\vec{s}) = b_n,
  \label{eq:PowerCountingOptimizedSystem}
\end{equation}
where the inhomogeneous term on the right-hand side is a consequence of the
inhomogeneous term, $w^{(l)[0]}$,  of \cref{eq:BasisAnsatz}.
In this way, the $y_m$ can be fixed by solving a simpler linear system than
\cref{eq:PowerCountingConstraints}. Importantly, $e$ is often very small and so
this represents an important speed up.

In summary, given a generating set of $\text{CSyz}(\Gamma)$, we construct a basis
of a space of surface terms $\text{CSyzSurface}^R(\Gamma, \vec{\nu}) \subset
\text{Surface}^R(\Gamma, \vec{\nu})$. The construction is performed for a given
numerical phase-space point, and optimized by structures learned from an initial
evaluation.

\subsubsection{Summary of Approach}\label{sec:algorithm}
Let us now summarize our approach for critical surface term generation. We use critical
syzygies to generate a set of surface terms within gluonic power counting
associated to a given diagram $\Gamma$ that are linearly independent on the
maximal cut associated to $\Gamma$.
By collecting these surface terms for all diagrams $\Gamma$, one then has the
necessary ingredients for integral reduction without raising propagator powers.
The resulting surface terms are implicitly stated as a set of analytic linear
systems that can be solved numerically on a given phase-space point.
The approach makes use of the Baikov representation, taking the Baikov
polynomial as input. In principle, this can be the Baikov polynomial of the top
topology, or the Baikov polynomial associated to the graph itself.
Our approach is composed of three major steps.
\begin{enumerate}
      \item Working on a numerical phase-space point, we compute the
            saturation index $\mu$ of the $J_{\text{syz}}^\Gamma$ with respect
            to the Baikov polynomial to check if it is $1$.
            We compute the dimension of $U_{\text{crit}[\log(B)]}^\Gamma$
            to check if it is zero. If either of these checks fail,
            we abort.\footnote{
    Study of the higher dimensional case is beyond the scope of this paper and
    left for future work. We discuss an ad-hoc solution to a single $\mu \ne 1$ case
    in \cref{sec:tthApplication}.
            }
      \item We construct a generating set of $\text{CSyz}(\Gamma)$ following the
            discussion of \cref{sec:CriticalSyzGenerators}. We fix an initial
            polynomial degree bound for the ansatz, $\vec{N}$, typically $N_j = 2$.
            We construct a tentative generating set with this degree bound and
            check if it generates $\text{CSyz}(\Gamma)$ via
            \cref{eq:CriticalGenerationTest}. If not, we
            raise the degree bounds in an ad-hoc manner until we successfully
            find a sufficiently high degree bound.
            The resulting generating set is presented as the solution to an
            analytic linear equation system, \cref{eq:criticalSyzygyEqsysRemoveLD}, which we
            optimize following the discussion in the latter part of \cref{sec:CriticalSyzGenerators}.
       \item Given this generating set of $\text{CSyz}(\Gamma)$, we construct a
             basis of $\text{CSyzSurface}^R(\Gamma, \vec{\nu})$ as discussed in
             \cref{sec:powercounting}. To maintain
             compact results, we employ two possible strategies to construct
             these surface terms. We first construct monomial multiples of our
             generating set. If such syzygies do not span the power counting
             space, we construct appropriate linear combinations via linear
             algebra.
             If this results in an incomplete set of surface terms, we return to
             step 2. We increase the degree bounds $\vec{N}$ and construct a
             new, larger generating set of $\text{CSyz}(\Gamma)$. We repeat this
             iteratively until an appropriate $\vec{N}$ has been found such that
             a complete set of surface terms is produced.
             The resulting basis of power-counting compatible surface terms is
             then either implicitly stated as an analytic linear system,
             \cref{eq:PowerCountingConstraints}, or as monomial multiples of
             $\text{CSyz}(\Gamma)$ generators.
\end{enumerate}
The output of this procedure is two fold.
First, there is an implicit representation of a generating set of
$\text{CSyz}(\Gamma)$ as a compact, analytic linear system in
\cref{eq:criticalSyzygyEqsysRemoveLD}.
Second, the basis of surface terms is either presented implicitly as a further
compact, analytic linear system, \cref{eq:PowerCountingConstraints} or as
explicit products of monomials in Baikov variables and $\text{CSyz}(\Gamma)$
generators.
The linear systems can then be solved to produce the necessary surface terms.
In practice, it is most fruitful to solve these systems numerically, as we do in
\cref{sec:tthApplication}.

Let us make a few comments on our approach. Firstly, we do not guarantee that it
produces a basis that spans $\text{CSyzSurface}^R(\Gamma, \vec{\nu})$. That is,
we do not guarantee that the isomorphism of \cref{eq:isomorphismWithPC} holds.
Nevertheless, we find in practical applications that it does.
Secondly, in step 3, we see that we cannot use an arbitrary generating set of
$\text{CSyz}(\Gamma)$ to construct a set of critical surface terms within power
counting by taking linear combinations. For this reason, we allow ourselves to
return to step 2 and increase the degree bounds.

\subsection{Application to Planar Contributions to $pp \rightarrow t\bar{t}H$ at Two Loops}\label{sec:tthApplication}
\begin{figure}[t]
    \centering
    \begin{minipage}{0.21\textwidth}
        \centering
        \includegraphics[width=\textwidth]{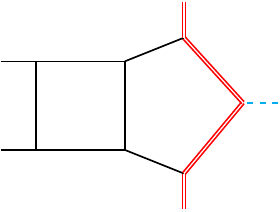}
    \end{minipage}
    \quad\,
    \begin{minipage}{0.21\textwidth}
        \centering
        \includegraphics[width=\textwidth]{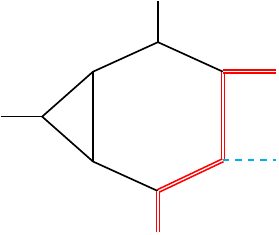}
    \end{minipage} 
    \quad\,
    \begin{minipage}{0.21\textwidth}
    \centering
    \includegraphics[width=0.9\textwidth]{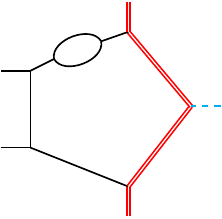}
    \end{minipage}
    \,\,
    \begin{minipage}{0.21\textwidth}
    \centering
    \includegraphics[width=0.95\textwidth]{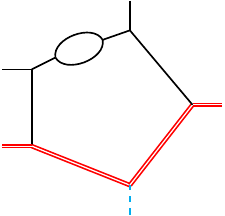}
    \end{minipage}
\\
\caption{ Four top sectors of leading color ttH production with light quark
  loop. Red lines denote the top quark, the blue dashed line denotes the Higgs boson. Black lines are massless
  particles (gluons or light quarks).
}
\label{fig:ggttHtopsector}
\end{figure}

In order to show that critical syzygies are a useful tool in practical
application to scattering processes, we apply our approach to the two-loop
leading-color contributions to the $pp \rightarrow
t\overline{t} H$ process.
Importantly, this allows us to study, in a physically relevant
case, the strength of the assumption that the saturation index, $\mu$, is 1 and
the critical locus of the logarithm of the Baikov polynomial,
$U^{\Gamma}_{\text{crit}(\log[B])}$, is a finite set of points.

We begin by applying the approach in \cref{sec:algorithm} to generate a
power-counting-compatible basis of surface terms for all diagrams in
leading-color light-quark-loop contributions in $t\overline{t}H$ production.
With a view to implementing these surface terms in the \texttt{Caravel}
framework, for each graph $\Gamma$ we work with the Baikov polynomial associated
to this graph (and not the top topology) as the remaining class of surface terms
are known analytically~\cite{Abreu:2017xsl}.
Diagrams depicting the four top-sectors are shown in \cref{fig:ggttHtopsector}.
We find that there are 123 non-factorizable subsectors which are inequivalent
under relabelling of external legs. As our representation of surface terms is
analytic, it is sufficient to analyze only these inequivalent sectors.
We compute the saturation index $\mu$ and the dimension of the critical locus of
the cut Baikov polynomial, using the computer algebra system \texttt{Singular}.
We perform the computation of both quantities on a numerical, finite-field
phase-space point, which leads to a negligible computation time.

Of the 123 inequivalent topologies, 122 give rise to a saturation index of
$\mu=1$ and a critical locus of $\log(B)$ on the maximal cut which is a finite
set of points. We consider these first.
Constructing the surface terms with our approach for each sector requires
solving linear equation system in \cref{eq:criticalSyzygyEqsysRemoveLD}. In
practice, we find that the largest linear equation system is of side-length 400.
In order to find a basis of surface terms within power counting, we must choose
a strategy for solving the power-counting constraints
\cref{eq:PowerCountingConstraints}. In most cases, we find it sufficient to use
the first strategy. Nevertheless, there are 9 sectors for which we apply the second
strategy.
In this case, after optimization, the most complicated linear system that we
must solve for the unknowns $y_m$ of \cref{eq:BasisAnsatz} contains 80 unknowns.
Having constructed the complete set of surface terms, we see that all
generating sets for the $\text{CSyz}(\Gamma)$ can be constructed with degree
bounds satisfying $N_j\leq 3$.

\begin{figure}[t]
    \centering
    \begin{minipage}{0.35\textwidth}
        \centering
        \includegraphics[width=\textwidth]{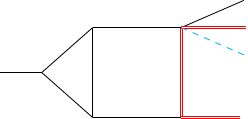}
    \end{minipage}
\caption{
  $t\overline{t}H$ topology for which the saturation index of
  $J_{\text{syz}}^\Gamma$ with respect to the Baikov polynomial is not 1.
  In this case, to recover the full set of surface terms, we allow for
  dimension shifted seed integrals.
  Black lines represent massless particles (i.e. gluons/light quarks), red lines
  represent the top quark and the blue dashed line represents the Higgs.
}
\label{fig:diagmu2}
\end{figure}

The single topology that has a saturation index of $\mu = 2$ is depicted in
\cref{fig:diagmu2}. We note that the critical locus of the logarithm of the
Baikov polynomial is still a finite set of points.
To compute the set of surface terms in this case, we take a two step approach.
First, we use the approach of \cref{sec:algorithm} to compute as many surface
terms as possible. In practice, we find that we miss only one surface term.
We then turn to the analysis of \cref{sec:NonTrivialSaturation} to construct the
remaining surface term. We first compute the saturation index
$\overline{\Delta}$ and find that it is $2$, telling us that this construction
does indeed produce a single new surface term in the large $\epsilon$, on-shell
limit.
To explicitly construct the surface term, we then solve
\cref{eq:DeltaProportionalDivergenceConstraint,eq:DeltaProportionalSyzConstraint}
by using a polynomial ansatz, successfully recovering the remaining surface term
and therefore constructing the full set.
We note that, due to the high degree of the Baikov polynomial, this is quite
computationally demanding. Nevertheless, the resulting surface term is quite simple.
It is interesting to also analyze this topology with the traditional Laporta approach.
We use LiteRed~\cite{Lee:2013mka} to generate IBP relations using
only single propagator powers as seeds and find that a single relation is also
missing.
However, if one allows for higher degree powers of the propagators in
seed integrals, then a complete reduction is observed.
This suggests that it would be interesting to study how to interpret seeds with
raised propagator powers in a critical-syzygy framework, a question we leave to
future work.

In order to check the validity of the surface terms that we produce, we
implement the analytic systems of linear equations that define the surface terms
into \texttt{Caravel}~\cite{Abreu:2020xvt}, which is able to use such systems to
perform numerical reductions of tensor integrals.
We then use \texttt{Caravel} to perform numerical reductions of tensor integrals
within power counting and check the validity of our implementation with
\texttt{FIRE 6.5}~\cite{Smirnov:2023yhb} at a series of randomly chosen
numerical phase-space points.

\begin{figure}[t]
    \centering
    $\eqnDiag{\includegraphics[scale=0.7]{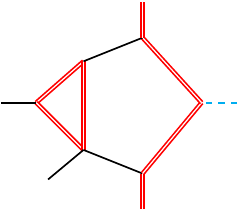}}$
    \,\,
    $\eqnDiag{\includegraphics[scale=0.7]{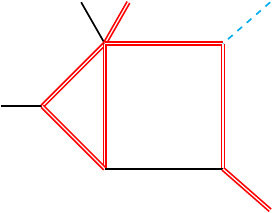}}$
    \,\,
    $\eqnDiag{\includegraphics[scale=0.7]{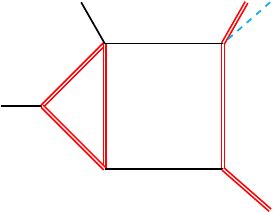}}$
    \,\,
    $\eqnDiag{\includegraphics[scale=0.7]{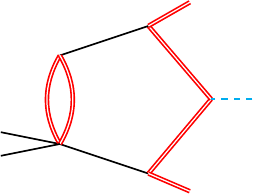}}$
    \\[1em]
    $\eqnDiag{\includegraphics[scale=0.8]{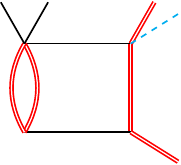}}$
    \qquad
    $\eqnDiag{\includegraphics[scale=0.8]{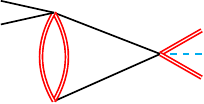}}$
    \qquad
    $\eqnDiag{\includegraphics[scale=0.8]{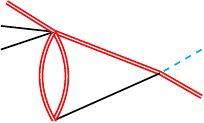}}$
    \qquad
    $\eqnDiag{\includegraphics[scale=0.8]{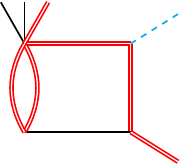}}$
  \caption{Collection of non-factorizable topologies from planar massive quark loop contributions
  to $pp \rightarrow t\overline{t}H$ for which the critical locus of $\log(B)$
  is not isolated.
  Red lines represent massive top quarks, black lines represent massless
  particles (gluons) and the blue line represents the Higgs.
  All integral topologies that contribute to $pp \rightarrow t\overline{t}H$
  which have a non-isolated critical locus are either permutations of those
  displayed, or are factorizable.}
  \label{fig:higherDimensional}

\end{figure}

Having considered the light-fermion-loop contributions to $pp \rightarrow
t\overline{t}H$, we next make an analysis of the pentabox topology that contributes
to the closed-top-loop contributions. Here, we find a number of topologies for
which the critical locus of the logarithm of the Baikov polynomial is not
a finite set of points. Many of these topologies are factorizable, containing the
one-loop integral of \cref{fig:OneLoopNonIsolated} as a factor, so this
observation is natural. Nevertheless, many of these integrals are not
factorizable, and we depict them in \cref{fig:higherDimensional}. For these
topologies, critical syzygies as studied in this paper are insufficient to
perform a complete reduction to master integrals, and we leave such a study to
further work. Interestingly, for each topology depicted in
\cref{fig:higherDimensional}, when considered in the Baikov representation
associated to corresponding graph, the critical locus on the maximal cut is only
one-dimensional. This unexpected simplicity hints that extension of the critical
syzygy formalism to cover these cases may be within reach.

\section{Summary and Outlook}
\label{sec:summary}

In this work, we have uncovered a new mathematical structure hidden within the
linear relations exhibited by Feynman integrals.
Working in the syzygy approach for constructing relations between
Feynman integrals and motivated by recent advances in intersection theory, we
considered how the numerators of integral relations, or ``surface terms'' behave
in the limit where the dimensional regulator, $\epsilon$, is taken to be large.
We showed how surface terms in the large-$\epsilon$ on-shell limit, must vanish
on critical points of the logarithm of the Baikov polynomial.
Moreover, we showed how this statement can be interpreted in the algebro-geometric
language of ideals, and how the ideals that arise are connected to the
Lee-Pomeransky approach for counting the number of master integrals.
This connection then motivated us to define a special class of
syzygies, which we dubbed ``critical syzygies''. Strikingly, while critical
syzygies are effectively a subset of the full syzygy module, for cases where the critical
locus of the Baikov polynomial is a finite set of points, we argued that they can be
used to construct the full set of surface terms in the large-$\epsilon$ limit.

In order to understand the practical construction of critical syzygies for
loop amplitude calculations, we made a number of studies at the one- and
two-loop level.
We first discussed how critical syzygies arise at one loop, providing an
alternative construction of OPP-like integrand bases.
We then moved to consider two-loop approaches, where we presented a
computational approach to construct critical syzygies.
We used this approach to study the two-loop example of planar Feynman integrals
for contributions to $pp \rightarrow t\overline{t}H$ production, directly
showing the applicability of critical syzygies to the light-fermion-loop case.
Interestingly, this allowed us to identify a case where careful study of the
multiplicity structure of the syzygies becomes important to obtain a complete
reduction to master integrals.

There are a number of further important avenues for work within the critical
syzygy approach. Firstly, critical syzygies are insufficient to perform IBP
reduction in cases where the critical locus of the cut Baikov polynomial is not
a finite set of points. A natural extension of our work would be to understand
if the critical syzygy framework could be extended to cover such cases.
Moreover, our study in \cref{sec:tthApplication} of the integral topology
relevant to $t\overline{t}H$ production where the multiplicity structure plays
an important role suggests that it would be useful to understand higher
propagator seeding in a critical syzygy framework.
Finally, the geometrical connection between critical syzygies and
critical/singular points of the Baikov polynomial motivates further work into
constructing analytical critical syzygy solutions.

\section*{Acknowledgments}
We thank Giulio Gambuti, Harald Ita, Pavel Novichkov and Vasily Sotnikov for insightful
discussions. We thank Harald Ita and Pavel Novichkov for comments on the draft.
The work of Qian Song was supported by the European Research Council (ERC) under
the European Union’s Horizon Europe research and innovation program grant
agreement 101078449 (ERC Starting Grant MultiScaleAmp). Views and opinions
expressed are however those of the authors only and do not necessarily reflect
those of the European Union or the European Research Council Executive Agency.
Neither the European Union nor the granting authority can be held responsible
for them.
\pagebreak

\appendix  
\section{Extended splitting lemma}
\label{app:splittingLemma}
Here we prove the lemma necessary to decompose $J_{\text{syz}}^\Gamma$ in \cref{sec:GeometryToAlgebra}.
This lemma can be regarded as an extended version of the splitting lemma
described in \cite[section B.3]{DeLaurentis:2022otd}.

\begin{lemma} Let \(R\) be a polynomial ring, \(J, K\) be ideals of \(R\) and \(b\) be an element of \(R\) 
such that \(Kb \subset J\). In this case one has that
\begin{equation}
    J = (J+K) \cap (J+b^\mu),
\label{eq:ExtendedSplittingLemma}
\end{equation}
where \(\mu\) is the saturation index of \(b\) with respect to \(J\).
\end{lemma}

\begin{proof}
It is clear that \(J \subseteq (J+K) \cap (J+b^\mu)\) as the intersection is of
two sets which each contain \(J\). Hence, if we have the reverse (non-proper)
inclusion, then we have equality. Let us consider an element of the
intersection. We shall prove that it is a member of \(J\), which will therefore prove
\eqref{eq:ExtendedSplittingLemma}.

Let us name the element in question \(c\). By definition we can write that
\begin{equation}
    c = j_1 + k = j_2 + t b^\mu,
\end{equation}
where \(j_i \in J\), \(k \in K\) and \(t \in R\). We aim to prove that \(t b^\mu \in J\). 
To do this, we shall consider
\begin{equation}
    t b^{\mu+1} = j_1 b + k b - j_2 b.
\end{equation}
Manifestly the right hand side is a sum of three elements of \(J\), and as \(J\) is
an ideal this implies that \(t b^{\mu+1}\) is also an element of \(J\). By the
definition of an ideal quotient, we have that this means that \(t \in J : b^{\mu+1}\).
However, as \(\mu\) is chosen to be the saturation index of \(b\), we have that
\(J:b^{\mu+1} = J:b^\mu\), which implies that \(t b^\mu \in J\). Looking back to the
definition of \(c\), it is clear that \(c \in J\) and therefore the equality in
\eqref{eq:ExtendedSplittingLemma} is proven.
\end{proof}

\bibliography{criticalsyz1}
\end{document}